\documentclass[a4paper, 10pt]{article}
\usepackage{jheppub}
\usepackage{bbm}
\usepackage{graphicx}
\usepackage{latexsym}
\usepackage{color}

\usepackage{amsmath,amsthm,amssymb,amsfonts,amscd}
\newtheorem{theorem}{Theorem}[section]

\numberwithin{equation}{section}

\def\ba#1\ea{\begin{align}#1\end{align}}
\def\bas#1\eas{\begin{align*}#1\end{align*}}
\newcommand{\mt}{\mapsto}
\newcommand{\bra}{\langle}
\newcommand{\ket}{\rangle}

\renewcommand{\Im}{\mathop{\rm Im}}
\newcommand\q{\quad}

\renewcommand{\d}{{\rm d}}
\newcommand{\e}{{\rm e}}

\newcommand\add{{\rm \bf add}}
\newcommand\sub{{\rm \bf sub}}
\newcommand\inve{{\rm \bf inv}}

\newcommand{\Nl}{\mathbb N}
\newcommand{\Zl}{\mathbb Z}

\newcommand{\Rl}{\mathbb R}
\newcommand{\Cl}{\mathbb C}

\newcommand{\ca}{\mathcal A}

\newcommand{\cc}{\mathcal C}

\newcommand{\ce}{\mathcal E}
\newcommand{\cf}{\mathcal F}

\newcommand{\ch}{\mathcal H}
\newcommand{\ci}{\mathcal I}

\newcommand{\cl}{\mathcal L}

\newcommand{\cp}{\mathcal P}

\newcommand{\cx}{\mathcal X}

     \newcommand{\fA}{\mathfrak{A}}

\renewcommand{\a}{\alpha}
\renewcommand{\b}{\beta}
\newcommand{\eps}{\epsilon}

\newcommand{\sig}{\sigma}

\newcommand{\w}{\omega}

\newcommand{\SO}{{\rm SO}}

\newcommand{\su}{\mathfrak{su}}
\newcommand{\SU}{{\rm SU}}

\newcommand{\U}{{\rm U}}

\title{On the space of generalized fluxes for loop quantum gravity\footnote{Published in: \href{http://stacks.iop.org/0264-9381/30/i=5/a=055008}{Class. Quantum Grav. 30 (2013) 055008} \copyright\, \href{http://iopscience.iop.org/page/copyright_notice}{IOP Publishing 2013}. This is an author-created, un-copyedited version of an article accepted for publication in Classical and Quantum Gravity. IOP Publishing Ltd is not responsible for any errors or omissions in this version of the manuscript or any version derived from it. The Version of Record is available online at \href{http://dx.doi.org/10.1088/0264-9381/30/5/055008}{doi:10.1088/0264-9381/30/5/055008}.}}

\author[a,b]{Bianca Dittrich,}
\author[a]{Carlos Guedes,}
\author[a]{Daniele Oriti}
\affiliation[a]{Max Planck Institute for Gravitational Physics, Am M\"uhlenberg 1, 14476 Golm, Germany}
\affiliation[b]{Perimeter Institute for Theoretical Physics, 31 Caroline St. N, Waterloo, ON N2L 2Y5, Canada}
\emailAdd{bianca.dittrich@aei.mpg.de, carlos.guedes@aei.mpg.de, daniele.oriti@aei.mpg.de}

\abstract{We show that the space of generalized fluxes -- momentum space -- for loop quantum gravity cannot be constructed by Fourier transforming the projective limit construction of the space of generalized connections -- position space -- due to the non-abelianess of the gauge group $\SU(2)$. From the abelianization of $\SU(2)$, $\U(1)^3$, we learn that the space of generalized fluxes turns out to be an inductive limit, and we determine the consistency conditions the fluxes should satisfy under coarse-graining of the underlying graphs. {We comment on the applications to loop quantum cosmology, in particular, how the characterization of the Bohr compactification of the real line as a projective limit opens the way for a similar analysis for LQC.}}

\keywords{Loop quantum gravity, loop quantum cosmology, cylindrical consistency, inverse limit, direct limit}

\begin{document}
\maketitle

\section{Introduction}
In standard quantum mechanics in flat space, the standard Fourier transform relates the two (main) Hilbert space representations in terms of (wave-)functions of position and of momentum, defining a duality between them.\footnote{In fact, since both manifolds, configuration space $G=\mathbb{R}^n$ and momentum space, defined as the Pontryagin dual $\widehat{G}=\mathbb{R}^n$, coincide, they are kinematically self-dual. However, the dynamics (e.g. Hamiltonian) may then differentiate between the two by not being symmetric under the exchange of position and momentum variables.} The availability of both is of course of practical utility, in that, depending on the system considered, each of them may be advantageous in bringing to the forefront different aspects of the system as well as for calculation purposes. Difficulties in defining a Fourier transform and a momentum space representation arise however as soon as the configuration space becomes non-trivial, in particular as soon as curvature is introduced, as in a gravitational context\footnote{The (restoration of) momentum-position duality, sometimes referred to as Born principle, has even been suggested as a guiding principle for the construction of theories of quantum gravity, particularly in the context of non-commutative geometry. See for example the recent work on the so-called ``relative locality" \cite{relativelocality}, and the earlier idea by Majid on ``co-gravity" \cite{cogravity}.}. No such definition is available in the most general case, that is in the absence of symmetries. On the other hand, in the special case of group manifolds or homogeneous spaces, when there is a transitive action of a group of symmetries on the configuration space, harmonic analysis allows for a notion of Fourier transform in terms of irreducible representations of the relevant symmetry group. This includes the case of phase spaces given by the cotangent bundle of a Lie group, when momentum space is identified with the corresponding (dual of the) Lie algebra, as it happens in loop quantum gravity \cite{TTbook}. However, a different notion of \textit{group Fourier transform} adapted to this group-theoretic setting has been proposed \cite{Freidel:2005me, Freidel:2005bb, Freidel:2005ec, Joung:2008mr, etera, matti}, and found several applications in quantum gravity models (see \cite{danielearistideSIGMA}). As it forms the basis of our analysis, we will introduce it in some detail in the following. Its roots can be traced back to the notions of quantum group Fourier transform \cite{Freidel:2005ec} and deformation quantization, being a map to non-commutative functions on the Lie algebra endowed with a star-product. The star-product reflects faithfully the choice of quantization map and ordering of the momentum space (Lie algebra) variables \cite{carlosmattidaniele}. As a consequence of this last point, observables and states in the resulting dual representation (contrary to the representation obtained by harmonic analysis) maintain a direct resemblance to the classical quantities, simplifying their interpretation and analysis.

\

Let us give a brief summary of the LQG framework. For more information about the intricacies of LQG refer to the original articles \cite{isham,Ashtekar:1993wf, Ashtekar:1994mh, Ashtekar:1994wa} or the comprehensive monograph \cite{TTbook}.
Loop quantum gravity is formulated as a symplectic system, where the pair of conjugate variables is given by holonomies $h_e[A]$ of an $\su(2)$-valued connection 1-form $A$ (Ashtekar connection) smeared along 1-dimensional edges $e$, and densitized triads $E$ smeared across 2-surfaces (electric fluxes). The smearing is crucial for quantization giving mathematical meaning to the distributional Poisson brackets, which among fundamental variables are
\bas
\{E^a_j(x),A^k_b(y)\}=\frac{\kappa}{2}\,\delta^a_b\,\delta^k_j\,\delta^{(3)}(x,y)\,,
\eas
where $a,b,c,\ldots$ are tangent space indices, and $i,j,k,\ldots$ refer to the $\su(2)$ Lie algebra. The same smearing leads to a definition of the classical phase space as well as of the space of quantum states based on graphs and associated dual surfaces. 

Since the smeared connection variables commute
\bas
\{h_e[A], h_{e'}[A]\}=0\,,
\eas
LQG is naturally defined in the connection representation. All holonomy operators can be diagonalized simultaneously and we thus have a functional calculus on a suitable space of generalized connections $\overline{\ca}$. 
 
A very important point is that despite the theory being defined on discrete graphs and associated surfaces, the set of graphs defines a directed and partially ordered set. Hence, refining any graph by a process called projective limit we are able to recover a notion of continuum limit. We will come back to this in the following, as we will attempt to define a similar continuum limit in terms of the conjugate variables.

Even though the triad variables Poisson commute, that is no longer the case for the flux variables:
\bas
E(S,f)=\int_S f^j\, E_j^a\, \eps_{abc} \,\d x^b\wedge \d x^c\,.
\eas
For instance, for operators smeared by two different test fields $f,g$ on the same 2-surface $S$ we have
\bas
\{E(S,f),E(S,g)\}\neq 0
\eas
if the Lie bracket $[f,g]^i=\eps^i_{jk}f^jg^k$ fails to vanish. Even smearing along two distinct surfaces, the commutator is again non-zero if the two surfaces intersect and the Lie bracket of the corresponding test fields is non-zero on the intersection. At first thought to be a quantization anomaly, in \cite{Ashtekar:1998ak} it was shown to be a feature that can be traced back to the classical theory. Thus, a simple definition of a momentum representation in which functions of the fluxes would act as multiplicative operators is not available. 

\

In the simplest case, for a given fixed graph, fluxes across surfaces dual to a single edge act as invariant vector fields on the group, and have the symplectic structure of the $\su(2)$ Lie algebra. Therefore, after the smearing procedure, the phase space associated to a graph is a product over the edges of the graph of cotangent bundles $T^*\SU(2)\simeq \SU(2)\times \su(2)^*$ on the gauge group.

\
 
Notwithstanding the fundamental non-commutativity of the fluxes, and taking advantage of the resulting Lie algebra structure and of the new notion of non-commutative Fourier transform mentioned above, in \cite{Baratin:2010nn}, a flux representation for loop quantum gravity was introduced. The work we are presenting here is an attempt to give a characterization of what the momentum space for LQG, defined through these new tools, should look like and how it can be constructed. 

\

Before presenting our results, let us motivate further the construction and use of such momentum/flux representation in (loop) quantum gravity (for an earlier attempt to define it, see \cite{lewandowski}). First of all, any new representation of the states and observables of the theory will in principle allow for new calculation tools that could prove advantageous in some situations. Most importantly, however, is the fact that a flux representation makes the geometric content of the same states and observables (and, in the covariant formulation, of the quantum amplitudes for the fundamental transitions that are summed over) clearer, since the fluxes are nothing else than metric variables. For the same reason, one would expect a flux formulation to facilitate the calculation of geometric observables and the coarse graining of states and observables with respect to geometric constraints \cite{coherentfluxes, coherentcollective} (we anticipate that the notion of coarse graining of geometric operators will be relevant also for our analysis of projective and inductive structures entering the construction of the space of generalized fluxes). The coupling of matter fields to quantum geometry is also most directly obtained in this representation \cite{johannesthomas}. Recently, new representations of the holonomy-flux algebra have been proposed for describing the physics of the theory around (condensate) vacua corresponding to diffeo-covariant, non-degenerate geometries \cite{HannoTimSigmaReview}, which are defined in terms of non-degenerate triad configurations, and could probably be developed further in a flux basis. 

 As flux representations encode more directly the geometry of quantum states, such representations might be useful for the coarse graining of geometrical variables \cite{improved,wroc,song}. In particular, \cite{ccd} presents a coarse graining in which the choice of representation and underlying vacuum is crucial.  

Finally, we recall that the flux representation has already found several applications in the related approaches of spin foam models and group field theories \cite{danielearistideSIGMA, aristidedanieleGFTmetric, aristidedanieleHolst, aristidedanieleGFTbarrett, GFTdiffeos, bubbles, bubblesjackets}, as well as in the analysis of simpler systems \cite{matti, mattidaniele}.

\

This flux representation was found by defining a group Fourier transform together with a $\star$-product on its image, first introduced in \cite{Freidel:2005me,Freidel:2005bb,Freidel:2005ec,Joung:2008mr} in the context of spin foam models. In this representation, flux operators act by $\star$-multiplication, and holonomies act as (exponentiated) translation operators. Using the projective limit construction of LQG, the group Fourier transform $\cf_\gamma$ was used to push-forward each level to its proper image, and in \cite{Baratin:2010nn} the following diagram was shown to commute
\ba
\begin{CD} 
\cup_\gamma\ch_\gamma @>\cf_\gamma>> \cup_\gamma \ch_{\star,\gamma}\\ 
@VV\pi V @VV \pi_{\star}V\\ 
\left(\cup_{\gamma}\ch_\gamma\right)/\sim @>\widetilde{\cf}>> \left(\cup_{\gamma}\ch_{\star,\gamma}\right)/\sim\\
\end{CD}
\label{eq:CD}
\ea
identifying the Hilbert space in the new triad representation as the completion of $\left(\cup_{\gamma}\ch_{\star,\gamma}\right)/\sim$. $\pi$ and $\pi_\star$ are the canonical projections with respect to the equivalence relation $\sim$ which is inherited from the graph structure (see section \ref{sec:nutshell}). In the connection representation we know that the (kinematical) Hilbert space is given by $\overline{\left(\cup_{\gamma}\ch_\gamma\right)/\sim}\simeq L^2(\overline{\ca},\d \mu_0)$, where $\overline{\ca}$ is the space of generalized connections and $\d\mu_0$ the Ashtekar-Lewandowski measure. Even if the Hilbert space in the triad representation can be defined by means of the projective limit, one would like to have a better characterization of the resulting space in terms of some functional calculus of generalized flux fields. Hence, the natural question is, can we write $\overline{\left(\cup_{\gamma}\ch_{\star,\gamma}\right)/\sim}\simeq L^2(\overline{\ce},\d\mu_{\star,0})$, for an appropriate space of generalized fluxes $\overline{\ce}$ and measure $\d\mu_{\star,0}$? This is precisely the issue we tackle in this paper.

\

We will see here that there are several obstructions to such a construction. First of all, when translating the projective limit construction on the connection side over to the image of the Fourier transform the notion of cylindrical consistency is violated whenever the gauge group is non-abelian. Thus, it is not possible to define the relevant cylindrically consistent $C^*$-algebra. This result is crucial since the space of generalized fluxes would arise as the corresponding spectrum\footnote{The spectrum of a $C^*$-algebra $\fA$ is the set of unitary equivalence classes of irreducible $^*$-representations or in the commutative case just the set of all non-zero $^*$-homomorphisms from $\fA$ to $\Cl$.} of this algebra. We note that even if it was possible to define such a cylindrically consistent $C^*$-algebra, the characterization that we are looking for would require a generalization of the Gel'fand representation theorem to noncommutative $C^*$-algebras, as the multiplication in the algebra is a noncommutative $\star$-product. Nevertheless, we can still learn something about the space of generalized fluxes by considering the abelianization of $\SU(2)$, that is, $\U(1)^3$. In fact, it has been shown that the quantization of linearized gravity leads to the LQG framework with $\U(1)^3$ as gauge group \cite{Varadarajan:2004ui}. It is even enough to work with one single copy of $\U(1)$, since the case $G=\U(1)^3$ is then simply obtained by a triple tensor product: not only the kinematical Hilbert space
\bas
\ch_{kin}^{\U(1)^3}=\ch_{kin}^{\U(1)}\otimes \ch_{kin}^{\U(1)} \otimes \ch_{kin}^{\U(1)}
\eas
has this simple product structure, but also the respective gauge-invariant subspaces decompose the same way \cite{Bahr:2007xa}.

\

The outline of the paper is the following: in the next section we briefly review the projective limit structure of LQG together with the notion of cylindrical consistency. Section \ref{sec:genflux} is the bulk of the paper. In \ref{app:noncylconv}, we start by showing the non-cylindrical consistency of the $\star$-product for non-abelian gauge groups, and proceed to the $\U(1)$ case. In subsection \ref{sec:gfft} we show that the space of generalized fluxes for $\U(1)$-LQG cannot be constructed as a projective limit, but in subsection \ref{sec:E_ind} we show how it arises as an inductive limit. The space of functions is then determined by pull-back giving rise to a suitable pro-$C^*$-algebra. In the conclusion \ref{sec:conc} we make some remarks on the analysis made here, and give an outlook on further work, in particular on the possibility of constructing a theory of loop quantum gravity tailored to the flux variables, and how the characterization of the Bohr compactification of the real line as a projective limit opens the way for a similar analysis for loop quantum cosmology.

\section{The notion of cylindrical consistency in a nutshell}
\label{sec:nutshell}
After identifying the space of generalized connections, i.e. the set $\overline{\ca}=\text{Hom}(\cp, G)$ of homomorphisms from the groupoid of paths $\cp$ to the group $G=\SU(2)$,  as the appropriate configuration space for loop quantum quantum gravity, the next step is to find the measure $\d\mu_0$ on this space to define the kinematical Hilbert space $\ch_0=L^2(\overline{\ca},\d\mu_0)$. This measure, called the  Ashtekar-Lewandowski measure, which is gauge and diffeomorphism invariant, is built by realizing $\ch_0$ as an inductive limit (also called direct limit) of Hilbert spaces $\ch_\gamma=L^2(\ca_\gamma,\d\mu_\gamma)$ associated to each graph $\gamma$ (for the definition of and more details on projective/inductive limits, refer to the \ref{app:projind}). The inductive structure is inherited by pullback from the projective structure in $\ca_\gamma=\text{Hom}(\overline{\gamma},G)$, where $\overline{\gamma}\subset\cp$ is the corresponding subgroupoid associated with $\gamma$. $\ca_\gamma$ is set-theoretically and topologically identified with $G^{|\gamma|}$, with $|\gamma|$ the number of edges in $\gamma$, and $\d\mu_\gamma$ defines the Haar measure on $G^{|\gamma|}$. Then, the identification of $\overline{\ca}$ with the projective limit (inverse limit) of $\ca_\gamma$ gives
\ba
\ch_0=\overline{\left(\cup_{\gamma}\ch_\gamma\right)/\sim}\,,
\label{lqg_hs}
\ea
where $\sim$ is an equivalence relation that determines the notion of cylindrical consistency (independence of representative) for the product between functions.

We remark that this result relies heavily on the Gel'fand representation theorem which states that any commutative $C^*$-algebra $\fA$ is isomorphic to the algebra of continuous functions that vanish at infinity over the spectrum of $\fA$, that is, $\fA\simeq C_0(\Delta(\fA))$.\footnote{More precisely, the Gel'fand representation theorem is an equivalence between the category of locally compact Hausdorff spaces and continuous proper maps and the opposite category of commutative $C^*$-algebras and proper $C^*$-morphisms.} The construction in \eqref{lqg_hs} is done at the level of $C(\ca_\gamma)$ in the sense that the spectrum of the commutative $C^*$-algebra $\overline{\cup_\gamma C(\ca_\gamma)/\sim}$ (the algebra of cylindrical functions) coincides with the projective limit of the $\ca_\gamma$'s. Compactness of $\ca_\gamma$ guarantees $C(\ca_\gamma)$ to be dense in $L^2(\ca_\gamma,\d\mu_\gamma)$. The existence of the measure is provided by the Riesz representation theorem (of linear functionals on function spaces) which basically states that linear functionals on spaces such as $C(\overline{\ca})$ can be seen as integration against (Borel) measures. The linear functional on $C(\overline{\ca})$ is then constructed by projective techniques through the linear functionals on $C(\ca_\gamma)$.

As already remarked, the existence of the projective limit guarantees the existence of a continuum limit of the theory despite it being defined on discrete graphs, at least at a kinematical level.\\

Since it will be important later on, let us describe in more detail the system of homomorphisms that give rise to the inductive/projective structure. Recall that the set of all embedded graphs in a (semi-) analytic manifold defines the index set over which the projective limit is taken. A graph $\gamma=(e_1,\ldots,e_n)$ is a finite set of analytic paths $e_i$ with 1 or 2-endpoint boundary (called edges), and we say that $\gamma$ is smaller/coarser than a graph $\gamma'$ (thus, $\gamma'$ is bigger/finer than $\gamma$), $\gamma\prec\gamma'$, when every edge in $\gamma$ can be obtained from a sequence of edges in $\gamma'$ by composition and/or orientation reversal. Then $(\ca_\gamma,\prec)$ defines a partially ordered and directed set, and we have, for $\gamma\prec\gamma'$, the natural (surjective) projections $p_{\gamma\gamma'}:\ca_{\gamma'}\rightarrow\ca_\gamma$ (restricting to $\ca_\gamma$ any morphism in $\ca_{\gamma'}$). These projections go from a bigger graph to a smaller graph and they satisfy
\bas
p_{\gamma\gamma'}\circ p_{\gamma'\gamma''}=p_{\gamma\gamma''}\,,\q \forall\, \gamma\prec\gamma'\prec\gamma''\,.
\eas
We thus have an inverse (or projective) system of objects and homomorphisms. These projections can be decomposed into three elementary ones associated to the three elementary moves from which one can obtain a larger graph from a smaller one compatible with operations on holonomies: (i) adding an edge, (ii) subdividing an edge, (iii) inverting an edge. See Figure \ref{fig:cyl}. 
\begin{figure}[t]
\centering
\includegraphics[scale=0.7]{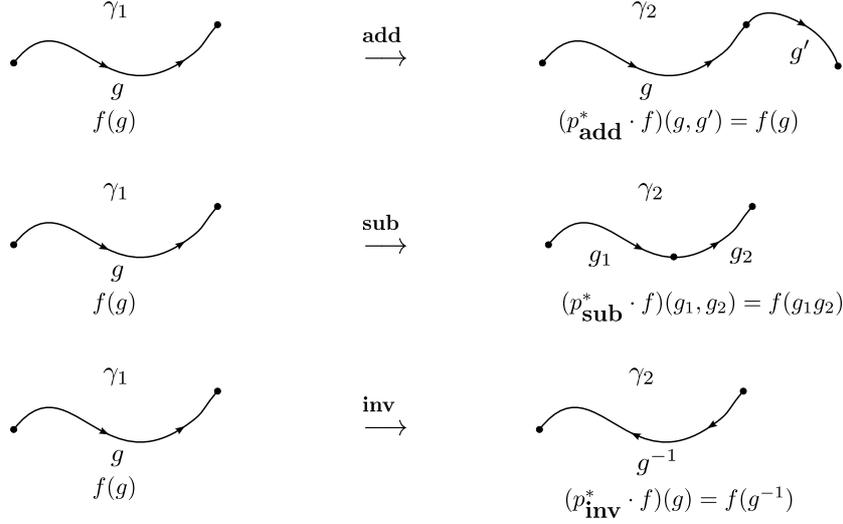}
\caption{The three elementary moves on graphs.}
\label{fig:cyl}
\end{figure}
Then,
\ba
& p_{\add}:\, \ca_{e,e'}\rightarrow \ca_e\,:\q\,\,\,\, (g,g') \mt g\nonumber\\
& p_\sub:\,\ca_{e_1,e_2}\rightarrow \ca_e\,;\q (g_1,g_2) \mt g_1g_2\nonumber\\
& p_\inve:\,\ca_e\rightarrow \ca_e\,;\qquad\;\;\;\;\;\;\;\;\;\;\; g \mt g^{-1}\,.
\label{eq:elementary}
\ea
The pullback of these defines the elementary injections for $\ch_\gamma$
\bas
 \add &:= p_\add^*:\, \ch_e\rightarrow \ch_{e,e'}\,;\q \,\,\,f(g)\mt (\add\cdot f)(g,g') =f(g)\\
 \sub &:= p_\sub^*:\, \ch_e\rightarrow \ch_{e_1,e_2}\,;\q  f(g)\mt (\sub\cdot f)(g_1,g_2) =f(g_1g_2)\\
 \inve &:= p_\inve^*:\, \ch_e\rightarrow \ch_e\,;\qquad\,\,\, f(g)\mt (\inve\cdot f )(g) =f(g^{-1})\,,
\eas
which determines the basic elements in $\ch_0$ in the same equivalence class, i.e. $[f]_\sim=\{f, \add\cdot f, \sub\cdot f, \inve\cdot f, \add\cdot\sub\cdot f,\add\cdot\inve\cdot f, \ldots\}$. Since $p^*_{\gamma\gamma'}: \ch_\gamma\rightarrow \ch_{\gamma'}$ go from a smaller graph to a bigger graph and satisfy
\bas
p^*_{\gamma'\gamma''}\circ p^*_{\gamma\gamma'}=p^*_{\gamma\gamma''}\,,\q \forall\, \gamma\prec\gamma'\prec\gamma''\,,
\eas
we have a direct (or inductive) system of objects and homomorphisms.

Let us check that the pointwise product in $\ch_0$ is indeed cylindrically consistent. Let $f,f'\in\ch_0$. By definition, we find graphs $\gamma,\gamma'$ and representatives $f_\gamma\in\ch_\gamma$, $f'_{\gamma'}\in\ch_{\gamma'}$ such that $f=[f_\gamma]_\sim$, $f'=[f'_{\gamma'}]_\sim$. Embed $\gamma,\gamma'$ in the common larger graph $\gamma''$, that is, $\gamma,\gamma'\prec\gamma''$. Then $f_{\gamma''}=p^*_{\gamma\gamma''}f_\gamma$, $f'_{\gamma''}=p^*_{\gamma'\gamma''}f'_{\gamma'}$, and $p^*_{\gamma\gamma''}f_\gamma=p^*_{\gamma'\gamma''}f_{\gamma'}$, $p^*_{\gamma\gamma''}f'_\gamma=p^*_{\gamma'\gamma''}f'_{\gamma'}$. Thus,
\bas
p^*_{\gamma\gamma''}(f_\gamma f'_\gamma)=p^*_{\gamma\gamma''}(f_\gamma)\, p^*_{\gamma\gamma''}(f'_\gamma)
=p^*_{\gamma'\gamma''}(f_{\gamma'})\, p^*_{\gamma'\gamma''}(f'_{\gamma'})=p^*_{\gamma'\gamma''}(f_{\gamma'}f'_{\gamma'})\,,
\eas
i.e. $f_\gamma f'_\gamma\sim f_{\gamma'}f'_{\gamma'}$, and the pointwise product does not depend on the representative chosen.
In terms of $\add,\sub,\inve$ this amounts to
\ba
\add\cdot (f\,f') & = (\add\cdot f)\, (\add\cdot f')\nonumber\\
\sub\cdot (f\,f') &= (\sub\cdot f)\, (\sub\cdot f')\nonumber\\
\inve\cdot (f\,f') &= (\inve\cdot f)\, (\inve\cdot f')\,,
\label{eq:cylconscond}
\ea
for $f,f'\in\ch_e$.

For a beautiful account on the structure of the space of generalized connections, refer to the article \cite{Velhinho:2004sf}.

\section{The space of generalized fluxes}
\label{sec:genflux}
This section constitutes the main part of the paper. The goal here is to define the analogue of the space of generalized connections $\overline{\ca}$ on the `momentum side' of the LQG phase space, that is the flux variables. The resulting space will be called the space of generalized fluxes $\overline{\ce}$. 

As said in the introduction, the natural approach for constructing such space fails for $G=\SU(2)$ (and any non-abelian group). The cylindrical consistency conditions used in defining the space of generalized connections are tailored to operations on holonomies ($h_{e_1\circ e_2}=h_{e_1}h_{e_2}$, $h_{e^{-1}}=h_{e}^{-1}$) and the non-abelianess of the group makes the translation to similar conditions on fluxes ill-defined. To show explicitly this difficulty is our first result. We are then constrained to work with the abelianization of $\SU(2)$, $\U(1)^3$, or rather $\U(1)$. Pushing-forward under the Fourier transform the projective limit construction of the space of generalized connections would lead to a definition of the space of generalized fluxes also as a projective limit. However, we will show that there is also an obstacle to this construction hitting again on the fact that the fluxes are significantly different from the connections.

Luckily, for $G=\U(1)$ the space of generalized connections is a true group opening in this way the possibility for a dual construction where the arrows are reversed. Thus, the projective limit is traded by an inductive limit and the previous problem disappears. The space of functions is finally defined by pull-back giving rise to a projective limit of $C^*$-algebras. Let us also note, at this point, that the $\U(1)$ case carries a further simplification, given by the fact that in this case the flux representation can be shown to be essentially equivalent to the charge network representation. We will clarify better in what sense this is true in section \ref{sec:E_ind}.

\subsection{The problems with $\SU(2)$}
\label{app:noncylconv}
Let us shortly summarize the commuting diagram from \cite{Baratin:2010nn}
\bas
\begin{CD} 
\cup_\gamma\ch_\gamma @>\cf_\gamma>> \cup_\gamma \ch_{\star,\gamma}\\ 
@VV\pi V @VV \pi_{\star}V\\ 
\left(\cup_{\gamma}\ch_\gamma\right)/\sim @>\widetilde{\cf}>> \left(\cup_{\gamma}\ch_{\star,\gamma}\right)/\sim\\
\end{CD}
\eas

For a single copy of $\SO(3)$ we define the noncommutative Fourier transform as the unitary map $\cf$ from $L^2(\SO(3),\d\mu_H)$, equipped with Haar measure $\d\mu_H$ (recently generalized to \SU(2) \cite{Dupuis:2011fx}), onto a space $L^2_\star(\Rl^3,\d\mu)$ of functions on $\su(2)\sim\Rl^3$ equipped with a noncommutative $\star$-product, and the standard Lebesgue measure:
\bas
\cf(f)(x)=\int_G \d g\,f(g)\,\e_g(x)\,,
\eas
where $\d g$ is the normalized Haar measure on the group, and $\e_g$ the appropriate plane-waves. The product is defined at the level of plane-waves as
\bas
\e_{g_1}\star \e_{g_2}=\e_{g_1g_2}\,,\q \forall\, g_1,g_2\in \SU(2)\,.
\eas
and extended by linearity to the image of $\cf$. As mentioned in the introduction, this non-commutative product is the result of a specific quantization map chosen for the Lie algebra part of the classical phase space \cite{carlosmattidaniele}.

The $\star$-product is crucial since it gives the natural algebra structure to the image of $\cf$, which is inherited from the convolution product in $L^2(\SO(3))$, that is, for $f,f'\in\ch_e$
\bas
\cf(f)\star\cf(f')=\cf(f * f')\,,
\eas
where the convolution product is as usual
\bas
(f * f')(g)=\int_G \d h\, f(gh^{-1})\, f'(h)\,.
\eas
We say that the $\star$-product is dual to convolution.

Extending to an arbitrary graph gives a family of unitary maps $\cf_\gamma:\ch_\gamma\rightarrow \ch_{\star,\gamma}$ labelled by graphs $\gamma$, where $\ch_\gamma:=L^2(\ca_\gamma,\d\mu_\gamma)\simeq L^2(\SO(3)^{|\gamma|},\d\mu_\gamma)$ and $\ch_{\star,\gamma}:=L^2_\star(\Rl^3)^{\otimes |\gamma|}$. Thus, we have the unitary map
\bas
\cf_\gamma:\, \ch_\gamma\rightarrow \ch_{\star,\gamma}\,,
\eas
and we want now to extend this to the full Hilbert space $\ch_0=\overline{\cup_\gamma\ch_\gamma/\sim}$. First, the family $\cf_\gamma$ gives a linear map $\cup_\gamma\ch_\gamma\rightarrow \cup_\gamma\ch_{\star,\gamma}$. In order to project it onto a well-defined map on the equivalence classes, we introduce the equivalence relation on $\cup_\gamma\ch_{\star,\gamma}$ which is `pushed-forward' by $\cf_\gamma$:
\bas
\forall\, u_{\gamma_i}\in \ch_{\star,\gamma_i}\,,\q u_{\gamma_1}\sim u_{\gamma_2}\q \Longleftrightarrow\q \cf_{\gamma_1}^{-1}(u_{\gamma_1})\sim \cf_{\gamma_2}^{-1}(u_{\gamma_2})\,.
\eas
That is, we have the injections $q_{\star,\gamma\gamma'}:\, \ch_{\star,\gamma}\rightarrow \ch_{\star,\gamma'}$ for all $\gamma\prec\gamma'$ defined dually by $q_{\star,\gamma\gamma'}\cf_\gamma:=\cf_\gamma\, p^*_{\gamma\gamma'}$. Using the definition, it is easy to see that they satisfy 
\ba
\label{inj_rel}
q_{\star,\gamma'\gamma''}\circ q_{\star,\gamma\gamma'}=q_{\star,\gamma\gamma''}\,,\q \forall\, \gamma\prec\gamma'\prec\gamma''\,,
\ea
i.e. we have an inductive system of objects and homomorphisms.

Finally, completion is now given with respect to the inner product pushed-forward by $\widetilde{\cf}$. That is, for any two elements $u,v$ of the quotient with representatives $u_{\gamma}\in\ch_{\star,\gamma}$ and $v_{\gamma'}\in\ch_{\star,\gamma'}$ the inner product is given by choosing a graph $\gamma''$ with $\gamma,\gamma'\prec \gamma''$ and elements $u_{\gamma''}\sim u_{\gamma}$ and $v_{\gamma''}\sim v_{\gamma'}$ in $\ch_{\star,\gamma''}$, and by setting
\bas
\langle u,v \rangle_{\widetilde{\cf}}:=\langle u_{\gamma''},v_{\gamma''}\rangle_{\cf_{\gamma''}}\,.
\eas
Since $\cf_\gamma$ are unitary transformations, the r.h.s. does not depend on the representatives $u_{\gamma}$, $v_{\gamma'}$ nor on the graph $\gamma''$. We thus have the complete definition of the full Hilbert space $\ch_{\star,0}=\overline{\left(\cup_{\gamma}\ch_{\star,\gamma}\right)/\sim}$ as an inductive limit. 

In \cite{Baratin:2010nn}, the question of cylindrical consistency of the star product in $\ch_{\star,0}$ was not posed and, as a Hilbert space, $\ch_{\star,0}$ makes perfect sense. However, to give the desired intrinsic characterization for $\ch_{\star,0}$ analogous to $\ch_0$, that is, to write $\ch_{\star,0}$ as $L^2(\overline{\ce},\d\mu_{\star,0})$ for some space of generalized fluxes $\overline{\ce}$ and measure $\d\mu_{\star,0}$, we need to make sure that $\ch_{\star,0}$ is well-defined as a $C^*$-algebra, in particular, that the $\star$-product is cylindrically consistent. As we have seen in section \ref{sec:nutshell} this amounts to the validity of Eqs.\,\eqref{eq:cylconscond} for the $\star$-product, or by duality, for the convolution product.

Then, for $\add$ we have
\bas
(\add\cdot (f * f')) (g,g')=(f * f') (g)\,,
\eas
and
\bas
((\add\cdot f) * (\add\cdot f'))(g,g') &=\int_{G^2}\d h \,\d h'\, (\add\cdot f)(gh^{-1},g'h'^{-1})\,(\add\cdot f')(h,h')\\
&=\int_{G^2} \d h\, \d h'\, f(gh^{-1})\, f'(h)=(f * f')(g).
\eas
So, it works for $\add$.

$\sub$ gives
\ba
(\sub\cdot (f * f')) (g,g')=(f * f') (gg')=\int_G \d h\, f(gg'h^{-1})\, f'(h)\,,
\label{convsub}
\ea
and
\bas
((\sub\cdot f) * (\sub\cdot f'))(g,g') &=\int_{G^2}\d h \,\d h'\, (\sub\cdot f)(gh^{-1},g'h'^{-1})\,(\sub\cdot f')(h,h')\\
&=\int_{G^2} \d h\, \d h'\, f(gh^{-1}g'h'^{-1})\, f'(hh')\\
&=\int_{G^2} \d h\, \d h'\, f(gh'h^{-1}g'h'^{-1}) f'(h)\,,
\eas
which matches \eqref{convsub} if and only if $G$ is abelian.

Lastly, $\inve$
\ba
(\inve\cdot (f * f')) (g)=(f * f') (g^{-1})=\int_G \d h\, f(g^{-1}h^{-1})\, f'(h)\,,
\label{convinve}
\ea
and
\bas
((\inve\cdot f) * (\inve\cdot f'))(g) &=\int_{G}\d h\, (\inve\cdot f)(gh^{-1})\,(\inve\cdot f')(h)\\
&=\int_{G} \d h\, f(hg^{-1})\, f'(h^{-1})\\
&=\int_{G} \d h\,  f(h^{-1}g^{-1}) f'(h)\,,
\eas
which again matches \eqref{convinve} if and only if $G$ is abelian.

Thus, the $\star$-product is not cylindrically consistent and consequently $\ch_{\star,0}$ is not a $C^*$-algebra. We emphasize that this result is independent of the specific Fourier transform used or of the specific form of the plane-waves, that is, any other quantization map chosen for the space of classical fluxes would have led to the same result. In order to have a well-defined algebra structure on the image of the Fourier transform we always need the multiplication to be dual to convolution, which, as we have just seen, is not cylindrically consistent unless the group $G$ is abelian. Let us also stress that similar issues would arise whenever one tries to define a kinematical continuum limit in variables dual to the connection and associated to surfaces. In particular, they would appear even using representation variables resulting from the Peter-Weyl decomposition, as they do in recent attempts to define refinement limits for the 2-complexes in the spin foam context \cite{refSF1,refSF2,refSF3,refSF4,refSF5,refSF6}.

As already said, the cylindrical consistency conditions are tailored to operations on holonomies, hence it is not too surprising that fluxes should not satisfy the same `gluing' conditions. Indeed, LQG kinematics treats connections and fluxes very asymmetrically. To understand better this asymmetry we will make the framework more symmetric by considering the abelianization of $\SU(2)$, $\U(1)^3$, where we can go further with the construction and still learn something about the space of generalized fluxes.

\subsection{The space of generalized fluxes by group Fourier transform: the abelian case}
\label{sec:gfft}

Loop quantum gravity with $\U(1)$ as gauge group is simpler in many aspects. In particular, the $\U(1)$ group Fourier transform and the $\star$-product reduce to the usual Fourier transform on the circle and the pointwise product, respectively. To avoid detouring too much from the main ideas of the text, we relegate to the \ref{sec:gft} an in-depth analysis of the $\U(1)$ group Fourier transform, where this is shown. Hence, $\cf$ is the unitary map
\bas
\cf(f)(x)=\frac{1}{2\pi}\int_{-\pi}^\pi \d\phi\, f(\phi)\,e^{-i\phi x}\,,
\eas
from $L^2(\U(1))\ni f$ onto $\ell^2(\Zl)$, the space of square-summable sequences (which has $C_0(\Zl)$ as a dense subspace), and the product on the image of $\cf$ is the usual pointwise product $(u\,v)(x)=u(x)\,v(x)$ for $u,v\in\ell^2(\Zl)$.
However, bear in mind that the $\U(1)$ group Fourier transform is fully defined on $\Rl$. That is, the conjugate variables to $\U(1)$ connections -- the fluxes -- are genuinely real numbers. It happens to be a feature of the $\U(1)$ group Fourier transform that it is sampled by its values on the integers\footnote{Notice the slightly analogous result for the $\SU(2)$ case, where a radial function on the image of the Fourier transform can be recovered from its values on the integers -- radial sampling theorem \cite{Freidel:2005ec}.} -- cf. \ref{sec:gft}.

The extension to an arbitrary graph and the projection onto the equivalence classes works out as in the previous subsection; the main difference is the abelian `$\star$-product' which now coincides with the pointwise product. It is still dual to convolution, but now that the group is abelian it is cylindrically consistent. We have the following result:

\begin{theorem}
$\ch_{\star,0}$ is a non-unital commutative $C^*$-algebra.
\label{thm:fyl_calg}
\end{theorem}
\begin{proof}
Strictly speaking, we are now looking at the Hilbert spaces $\ch_{\star,\gamma}=\ell^2(\Zl^{|\gamma|})$ at the algebraic level $C_0(\Zl^{|\gamma|})$ (which form dense subspaces). Each of the spaces $C_0(\Zl^{|\gamma|})$ is a non-unital commutative $C^*$-algebra with respect to complex conjugation, $\sup$-norm, and pointwise multiplication. Then, it just remains to check that the operations on the full algebra $\cup_\gamma C_0(\Zl^{|\gamma|})/\sim$, such as the product and the norm do not depend on the representative in each equivalence class, i.e. they are cylindrically consistent. 

For the ($\star$-)product this amounts to
\bas
(\add_\star\cdot u)\, (\add_\star\cdot v) &=\add_\star\cdot(u\, v)\,,\\
(\sub_\star\cdot u)\, (\sub_\star\cdot v) &= \sub_\star\cdot(u\, v)\,,\\
(\inve_\star\cdot u)\, (\inve_\star\cdot v) &= \inve_\star\cdot(u\, v)\,.
\eas
The action of $\add$, $\sub$, and $\inve$ is given below, Eq.\,\eqref{ind_func}. Explicitly, for $\add$, we have
\bas
((\add_\star\cdot u)\, (\add_\star\cdot v))(x_1,x_2) &=(\add_\star\cdot u)(x_1,x_2)\,(\add_\star\cdot v)(x_1,x_2)=u(x_1)\,\delta_{0,x_2}\,v(x_1)\,\delta_{0,x_2}\\ &=(u\, v )(x_1)\,\delta_{0,x_2}=(\add_\star\cdot (u\, v))(x_1,x_2)\,.
\eas
$\sub$ and $\inve$ can be shown similarly.

The norm
\bas
||u||=\sup_{\mathbf{x}\in \Zl^n}|u(\mathbf{x})|
\eas
satisfies
\bas
||\add_\star\cdot u||=||\sub_\star\cdot u||=||\inve_\star\cdot u||=||u||\,,
\eas
and is thus also well-defined.

Hence, $\cup_\gamma C_0(\Zl^{|\gamma|})/\sim$ is a non-unital commutative $C^*$-algebra.
\end{proof}

The definition of $\ch_{\star,0}$ as an inductive limit of abelian $C^*$-algebras sets us \emph{almost} on the same footing as the standard kinematical Hilbert space for loop quantum gravity $\ch_0$. The method used to determine the spectrum of the $C^*$-algebra relies heavily on the fact that if $(\fA_\a,p^*_{\a\b},\ci)$ is an inductive family of abelian $C^*$-algebras $\fA_\a$, where $\ci$ is a partially ordered index set, the inductive limit $\fA$ is a well-defined abelian $C^*$-algebra whose spectrum $\Delta(\fA)$ is a locally compact Hausdorff space homeomorphic to the projective limit of the projective family $(\Delta(\fA_\a),p_{\a\b},\ci)$.

As for the space of cylindrical functions, the inductive system of homomorphisms splits into three elementary ones defined dually by $q_{\star,\gamma\gamma'}:=\cf_\gamma\, p^*_{\gamma\gamma'}\cf_\gamma^{-1}$:
\ba
\add_\star &:= q_{\star,\add}:\, \ch_{\star,e}\rightarrow \ch_{\star,e,e'}\,,\q\,\,(\add_\star\cdot u)(x_1,x_2)  :=(\cf(\add\cdot f))(x_1,x_2)=u(x_1)\,\delta_{0,x_2}\,,\nonumber\\
\sub_\star &:= q_{\star,\sub}:\, \ch_{\star,e}\rightarrow \ch_{\star,e_1,e_2}\,,\q(\sub_\star\cdot u)(x_1,x_2) := (\cf(\sub\cdot f))(x_1,x_2)=u(x_1)\,\delta_{x_1,x_2}\,,\nonumber\\
\inve_\star &:= q_{\star,\inve}:\, \ch_{\star,e}\rightarrow \ch_{\star,e}\,,\qquad\;\;\;\;\;\;\;\;\; (\inve_\star\cdot u)(x) :=(\cf(\inve\cdot f))(x)=u(-x)\,,
\label{ind_func}
\ea
which again determine the elements in $\cup_\gamma\ch_{\star,\gamma}/\sim$ in the same equivalence class, i.e. $[u]_\sim=\{u, \add_\star\cdot u,\sub_\star\cdot u, \inve_\star\cdot u, \ldots \}$. 

Eqs.\,\eqref{ind_func} define our inductive system of functions through the injections $q_{\star,\gamma\gamma'}$. Recall that the usual procedure for LQG starts with the projections \eqref{eq:elementary}, and the injections at the level of functions are simply defined by pullback. Here we already have the system of injections \eqref{ind_func} and, should they exist, we want to determine the system of projections $p_{\star,\gamma\gamma'}$ that give rise to these injections. That is, are the injections $q_{\star,\gamma\gamma'}$'s the pullback of some projections $p_{\star,\gamma\gamma'}$'s: $q_{\star,\gamma\gamma'}=p_{\star,\gamma\gamma'}^*$? For the three elementary operations, we are looking for projections $p_{\star,\add}$, $p_{\star,\sub}$, and $p_{\star,\inve}$ such that
\bas
(p_{\star,\add}^*\, u)(x_1,x_2) &\equiv u(p_{\star,\add}(x_1,x_2))=u(x_1)\,\delta_{0,x_2}\,, \\
(p_{\star,\sub}^*\, u)(x_1,x_2) &\equiv u(p_{\star,\sub}(x_1,x_2))=u(x_1)\,\delta_{x_1,x_2}\,,\\
(p_{\star,\inve}^*\, u)(x) &\equiv u(p_{\star,\inve}(x))=u(-x)\,,
\eas
holds.

Using the fact that $u\in C_0$ and thus vanish at infinity, we can naively define the projections as
\ba
p_{\star,\add}(x_1,x_2) &:= \begin{cases}
x_1 & \text{if}\q x_2=0\\
\infty & \text{if}\q x_2\neq 0
\end{cases}\,,\nonumber \\
p_{\star,\sub}(x_1,x_2) &:= \begin{cases}
x_1 & \text{if}\q x_1=x_2\\
\infty & \text{if}\q x_1\neq x_2
\end{cases}\,,\nonumber\\
p_{\star,\inve}(x) &:= -x\,.\nonumber \\
\label{fyl_cons}
\ea
However, the above is rather formal and one runs into several technical problems in trying to justify the use of the infinity as an element of the target space of the projections. First of all, infinity does not belong to $\Zl$, so strictly speaking \eqref{fyl_cons} does not define a map. One could consider the one-point compactification (or Alexandroff extension) of the integers but this amounts to change the algebra itself (as it now becomes unital) and functions will not vanish at infinity anymore. One could try to make the limiting procedure precise by using the very definition of $u\in C_0(\Zl^{|\gamma|})$, that is, since $\Zl^{|\gamma|}$ is locally compact, there exists a compact set $K\subseteq \Zl^{|\gamma|}$ such that $|u(x)|<\eps$ for every $\eps>0$ and for every $x\in \Zl^{|\gamma|}\backslash K$. But this means that the whole procedure would only allow an indirect characterization of the underlying space (of generalized fluxes) through the behaviour of the space of functions. As a result, any such definition would fail to provide the intrinsic characterization of $\overline{\ce}$ that we are looking for.

Even though we do not have a proof that such projections $p_{\star,\gamma\gamma'}$ do not exist, it seems rather unnatural to force such a construction since the structure of connections and fluxes is significantly different. Indeed, one is the Fourier transform of the other and in the Fourier transform `arrows' are naturally reversed. In particular projections are changed into inductions and vice versa, at least in this abelian case. Indeed, we will see in the next subsection how reversing the arrows in the categorical sense makes it possible to define the space of generalized fluxes as an inductive limit\footnote{It is still not excluded that there might exist other methods of determining the spectrum which do not require the construction of the projective limit. If we could make sense of the projections and if the projective limit (of the discrete spaces $\Zl^{|\gamma|}$) could be defined, we would end up with a totally disconnected Hausdorff space, also known as complete ultrametric or non-Archimedean spaces.}.

\subsection{The space of generalized fluxes by duality}
\label{sec:E_ind}
The framework of $\U(1)$-LQG provides a different strategy for determining the space of generalized fluxes. The crucial point here is that for $G=\U(1)$ the space of generalized connections $\overline{\ca}$ is a true group, and the following theorem, giving the natural way of trading a projective system by an inductive system, is applicable. 
\begin{theorem}
Suppose $\ca_\gamma$ are abelian groups, and let $\overline{\ca}$ be the projective limit with projections $p_\gamma: \overline{\ca}\rightarrow \ca_\gamma$. Then, the dual group $\widehat{\overline{\ca}}$ equals the inductive limit of the dual groups $\widehat{\ca}_\gamma$.
\label{thm:proj-ind}
\end{theorem}
\begin{proof}
Let $p_\gamma:\overline{\ca} \rightarrow \ca_\gamma$ be the projections. Then, 
\bas
\hat{p}_\gamma:\, & \widehat{\ca}_\gamma \rightarrow \widehat{\overline{\ca}}\\ 
 &\chi_\gamma  \mapsto \hat{p}_\gamma(\chi_\gamma)\,,
\eas 
such that $\hat{p}_\gamma(\chi_\gamma)(g):=\chi_\gamma(p_\gamma(g))$, $g\in\overline{\ca}$, defines the morphisms in the dual system (direct).
In particular, the inverse system of mappings
\bas
p_{\gamma\gamma'}: \ca_{\gamma'}\rightarrow \ca_{\gamma}\,,\q \gamma\prec\gamma'\,,
\eas
where, for all $\gamma\prec\gamma'\prec\gamma''$, satisfy $p_{\gamma\gamma'}\circ p_{\gamma'\gamma''}=p_{\gamma\gamma''}$, gives rise to the corresponding direct system of mappings
\bas
\hat{p}_{\gamma\gamma'}: \widehat{\ca}_{\gamma}\rightarrow \widehat{\ca}_{\gamma'}\,,
\eas
where $\hat{p}_{\gamma\gamma'}(\chi_\gamma)(g_{\gamma'}):=\chi_\gamma(p_{\gamma\gamma'}(g_{\gamma'}))=\chi_\gamma(g_\gamma)$, for $g_{\gamma'}\in \ca_{\gamma'}$. Using the associativity for the inverse system, it is straightforward to show that
\bas
\hat{p}_{\gamma'\gamma''}\circ \hat{p}_{\gamma\gamma'}=\hat{p}_{\gamma\gamma''}\,,\q \forall\,\gamma\prec\gamma'\prec\gamma''\,,
\eas
that is, the mappings $\hat{p}_{\gamma\gamma'}$ do indeed define an inductive system.
\end{proof}

We are now in position to determine the sought for dual construction, that is, the inductive system. Recall that $\ca_\gamma$ may be identified with $\U(1)^{|\gamma|}$ and the Pontryagin dual is just $\widehat{\ca}_\gamma=\Zl^{|\gamma|}$ through the identification $\chi_{x_1,\ldots,x_\gamma}(z_1,\ldots, z_\gamma)=z_1^{x_1}\cdots z_\gamma^{x_\gamma}$, for $z_1,\ldots,z_\gamma\in \U(1)$ and $x_1,\ldots,x_\gamma\in \Zl$. Therefore, Eqs.\,\eqref{eq:elementary} give
\bas
& \hat{p}_{\add}:\, \Zl \rightarrow \Zl^2\,,\q x\mt \hat{p}_\add(x)\,,\\
& \hat{p}_{\sub}:\, \Zl \rightarrow \Zl^2\,,\q x\mt \hat{p}_\sub(x)\,,\\
& \hat{p}_{\inve}:\, \Zl \rightarrow \Zl\,,\q\;\; x\mt \hat{p}_\inve(x)\,,
\eas
whose actions on $z_1,z_2\in \U(1)$ are, respectively,
\bas
& \hat{p}_\add(\chi_x)(z_1,z_2)=\chi_x(p_\add(z_1,z_2))=\chi_x(z_1)=z_1^x=z_1^x\,z_2^0=\chi_{x,0}(z_1,z_2)\,,\\
& \hat{p}_\sub(\chi_x)(z_1,z_2)=\chi_x(p_\sub(z_1,z_2))=\chi_x(z_1\,z_2)=(z_1\,z_2)^x=z_1^x\,z_2^x=\chi_{x,x}(z_1,z_2)\,,\\
& \hat{p}_\inve(\chi_x)(z_1)=\chi_x(p_\inve(z_1))=\chi_x(z_1^{-1})=z_1^{-x}=\chi_{-x}(z_1)\,.
\eas

Thus, the embeddings are simply
\ba
& \hat{p}_{\add}:\, \Zl \rightarrow \Zl^2\,;\q x \mt (x,0)\,,\nonumber\\
& \hat{p}_{\sub}:\, \Zl \rightarrow \Zl^2\,;\q x \mt (x,x)\,,\nonumber\\
& \hat{p}_{\inve}:\,\Zl \rightarrow \Zl\,;\q\;\; x \mt -x\,,
\label{ind_labels}
\ea
and have a very nice flux interpretation which agrees with our intuition of how fluxes should behave under coarse-graining of the underlying graph: (i) adding an edge should not bring more information into the system, so the flux on the added edge is zero, (ii) subdividing an edge does not change anything and thus the flux through the subdivided edges is the same, (iii) inverting an edge just changes the direction of the flux, picking up a minus sign. See Figure \ref{fig:cyl_surf}.
\begin{figure}[t]
\centering
\includegraphics[scale=0.7]{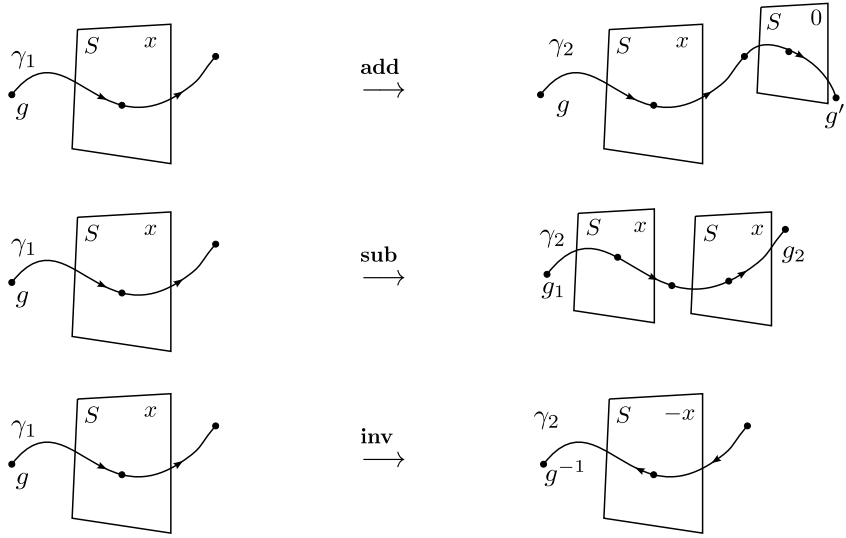}
\caption{Consistency conditions for fluxes across surfaces associated with the three elementary moves on graphs.}
\label{fig:cyl_surf}
\end{figure}

Using Theorem \ref{thm:proj-ind} we know that Eqs.\,\eqref{ind_labels} define an inductive system. Hence, we may define the space of generalized fluxes for $\U(1)$-LQG as the inductive limit of $\Zl^{|\gamma|}$'s which agrees with the Pontryagin dual $\overline{\ce}=\widehat{\overline{\ca}}=\widehat{\text{Hom}(\cp,\U(1))}=\text{Hom}(\text{Hom}(\cp,\U(1)),\Cl)$. The important point here is not the explicit form, which is not very enlightening, but the fact that it can be defined consistently as an inductive limit and, above all, the gluing conditions \eqref{ind_labels} from which it arises.

To finish this section we define the corresponding space of functions. The pullback of the embeddings \eqref{ind_labels} gives the following projections
\ba
(\hat{p}^*_{\add} \cdot u) (x) &=u(x,0)\,,\nonumber\\
(\hat{p}^*_{\sub} \cdot u)(x) &=u(x,x)\,,\nonumber \\
(\hat{p}^*_{\inve} \cdot u)(x) &=u(-x)\,,
\label{eq:dual_rel}
\ea
which give the consistency conditions to define the projective limit of the $C^*$-algebras $C_0(X_\gamma)$ for $X_\gamma=\Zl^{|\gamma|}$. Notice that the partial order for $X_\gamma$ induces the same partial order for $C_0(X_\gamma)$. An element $(u_\gamma)_\gamma$ of the projective limit is an element of the product $\times_\gamma C_0(\Zl^{|\gamma|})$ subject to the conditions $\hat{p}^*_{\gamma\gamma'}(u_{\gamma'})=u_\gamma$ for $\gamma\prec\gamma'$, and so a quite complicated object.

A projective limit of $C^*$-algebras goes by the name pro-$C^*$-algebra (also known as $LMC^*$-algebra, locally $C^*$-algebra or $\sigma$-$C^*$-algebra).
The Gel'fand duality theorem can be extended to commutative pro-$C^*$-algebras and from the perspective of non-commutative geometry, pro-$C^*$-algebras can be seen as non-commutative $k$-spaces \cite{procstar}.  Let us also remark that pro-$C^*$-algebras are in general not Hilbert spaces, although they might contain Hilbert subspaces. Therefore, they possess much more information than usual Hilbert spaces, as we detail in the next subsection.\\

Let us, at this point, emphasize that in this $\U(1)$ case, the fluxes can be identified with the charge network basis, since in a sense their only relevant component is the modulus which corresponds to the charge. However, their modulus remains valued in the real numbers, as opposed to what we would have would the flux representation and charge network basis be exactly the same. What happens next, due to the sampling mentioned at the beginning of \ref{sec:gfft}, is that the functions are fully specified by the evaluation on the integers and, therefore, in this simple $\U(1)$ case, working in the flux representation is fully equivalent to working with the, a priori different, charge network representation.

\subsection{More on projections and inductions}
Let us give more detail on the system of projections and inductions defined above and on the Hilbert spaces we (would try to) construct from them. In particular, we clarify here in which sense the pro-$C^*$-algebra constructed above from the projective system is much bigger than the usual Hilbert space.

\

Projective and inductive limit are in general related by duality. That is, a projective system of labels induces an inductive system of functions simply by pullback, and vice-versa. However, they can be both defined independently as is done in the \ref{app:projind}. Here, we deal with a system of projections $\pi_{\gamma\gamma'}$ and a system of inductions $\iota_{\gamma'\gamma}$ for one and the same space. First of all, note the following relation
\ba
\pi_{\gamma\gamma'}\circ\iota_{\gamma'\gamma}=\text{id}_\gamma\,,
\label{eq:iotapi}
\ea
for any pair of graphs $\gamma\prec\gamma'$. It is straightforward to check that the inductions \eqref{ind_func} and the projections \eqref{eq:dual_rel} satisfy \eqref{eq:iotapi}. Recall that $\pi_{\gamma\gamma'}:X_{\gamma'}\rightarrow X_\gamma$, $\iota_{\gamma'\gamma}:X_{\gamma}\rightarrow X_{\gamma'}$, for some collection of objects $\{X_\gamma\}_{\gamma\in\cl}$ and $\cl$ a directed poset. We remark that the reverse equation $\iota_{\gamma'\gamma}\circ\pi_{\gamma\gamma'}=\text{id}_{\gamma'}$ does not generically hold, since when we first project and then embed we are typically throwing some information away. 
 
With the inductions at hand one may define the inductive limit, while with the projections one may define the projective limit. Using \eqref{eq:iotapi} we will see how one can understand an element of the inductive limit in terms of an element of the projective limit, however, \textit{not} vice-versa. Recall from \ref{app:projind} that elements of the projective limit $X_{\text{proj}}$ are nets (i.e. elements in the direct product over all graphs) subjected to a consistency condition
\bas
X_{\text{proj}}=
\{(x_\gamma)\in \prod X_\gamma \q |\q \pi_{\gamma\gamma'}(x_{\gamma'})=x_{\gamma}\q \text{for all}\q \gamma\prec\gamma'\}\,.
\eas
While elements of the inductive limit consist of equivalence classes of elements of the disjoint union over all graphs
\bas
X_{\text{ind}}=\cup_\gamma X_\gamma / \sim\,,
\eas
where $x_\gamma\sim x_{\gamma'}$ means that there exists $\gamma''$ such that $\iota_{\gamma''\gamma}(x_\gamma)=\iota_{\gamma''\gamma'}(x_{\gamma'})$ and $\gamma,\gamma'\prec\gamma''$.
Now, given an element $y=[y]_\sim\in X_{\text{ind}}$ we will construct an element $x\in X_{\text{proj}}$, i.e. an assignment $\gamma\mapsto x_\gamma$ for all graphs $\gamma$, such that $x_{\gamma'}=y_{\gamma'}$ for all $\gamma'$ for which a representative $y_{\gamma'}$ in the equivalence class $y$ exists. In this sense we can embed the inductive limit into the projective limit.\\

Pick some element $y_\gamma$ in the equivalence class $y$. Then for any graph $\gamma'$ define the assignment $\gamma\mapsto x_\gamma$ as follows: choose $\gamma''$ such that $\gamma,\gamma'\prec\gamma''$, then
\bas
x_{\gamma'} =\pi_{\gamma'\gamma''}\circ \iota_{\gamma''\gamma} (y_\gamma)\,,
\eas
gives a consistent definition of an element of the projective limit, is independent of the choice $y_\gamma$ in $y$ and moreover $x_\gamma$ agrees on all graphs on which a representative $y_\gamma$ of $y$ exists.

In this way, we can map the inductive limit into the projective limit, however not surjectively. The image of the inductive limit consists of elements $x$ for which there exists a graph $\gamma''$ such that $\iota_{\gamma''\gamma}\circ\pi_{\gamma\gamma'}(x_{\gamma'})=x_{\gamma''}$ for all $\gamma,\gamma'$ with $\gamma,\gamma'\prec \gamma''$. The existence of such a `maximal graph' $\gamma''$ is however not guaranteed for generic elements of the projective limit.
For this reason we cannot use the Hilbert space structure of the inductive limit to make the projective limit into a Hilbert space, confirming the fact that pro-$C^*$-algebras are much bigger objects than Hilbert spaces.

\section{Conclusion and Outlook}
\label{sec:conc}

The loop quantum gravity kinematics treats connections and fluxes very asymmetrically. Therefore, the projective limit construction of the space of generalized connections does not translate trivially over to the flux side. Let us summarize what we have done. 

Using the commuting diagram \eqref{eq:CD} from \cite{Baratin:2010nn} we set out to give an intrinsic characterization of the space $\ch_{\star,0}=\left(\cup_{\gamma}\ch_{\star,\gamma}\right)/\sim$ in terms of some functional calculus of generalized flux fields. We have seen that the $\star$-product on the image of the Fourier transform is not cylindrically consistent unless the gauge group $G$ is abelian, and consequently $\ch_{\star,0}$ is not a $C^*$-algebra. This result is important because it means we cannot make sense of the space of generalized fluxes as the spectrum of this (would be) algebra. In more physical terms, this result suggests that a definition of the continuum limit of the theory, even at the pure kinematical level, cannot rely on the cylindrical consistency conditions coming from operations on holonomies, if one wants it to imply also a continuum limit in the dual flux representation, and thus coming from a  proper coarse-graining of fluxes. Rather, it seems to suggest that a new construction is needed.

Nevertheless, we were still able to learn something about the space of generalized fluxes by considering the gauge group $\U(1)$. In this setting we found out that the space of generalized fluxes cannot be constructed as a projective limit, but arises naturally as an inductive limit. The cylindrical consistency conditions for the fluxes \eqref{ind_labels} turned out to have a very nice physical interpretation (see Figure \ref{fig:cyl_surf}): (i) adding an edge should not bring more information into the system, so the flux on the added edge is zero, (ii) subdividing an edge does not change anything and thus the flux through the subdivided edges is the same, (iii) inverting an edge just changes the direction of the flux, picking up a minus sign. Even though we determined the space of generalized fluxes for $\U(1)$-LQG to be $\overline{\ce}=\text{Hom}(\text{Hom}(\cp,\U(1)),\Cl)$, one would like to have a better description of this space. Finally, the space of functions was defined by pull-back giving rise to a projective limit of $C^*$-algebras. We showed that this algebra is in general much bigger than the usual Hilbert space but, it might still be possible to improve its characterization by noting that the Gel'fand duality theorem can be extended to pro-$C^*$-algebras \cite{procstar}.

In light of our results, we conclude with an outlook on two issues worth pursuing further.

\begin{description}
\item[LQG from scratch and coarse-graining of fluxes.] Loop quantum gravity as it is formulated is entirely based on graphs. As we have seen the cylindrical consistency conditions do not translate easily to the flux side, specially because they are tailored to operations on holonomies. Therefore it seems misguided to force them on the flux variables. In the abelian case we learned that the fluxes compose according to \eqref{ind_labels}, with the aforementioned natural geometric interpretation. However, this process does not correspond to a coarse-graining in the same way as the family of projections does for the holonomies: surfaces are added according to how their dual edges compose, i.e. the operation of `adding' puts two surfaces `parallel' to each other -- but does not add them into a bigger surface. The question arises whether one can come up with a family of inductions that would rather represent these geometrically natural coarse-grainings.
This direction of thoughts hits however on many difficulties encountered before: one, is the more complicated geometrical structure of surfaces as compared to edges; another, is that for a gauge covariant coarse-graining of fluxes one not only needs the fluxes but also the holonomies to parallel transport fluxes. To avoid these difficulties one could consider again Abelian groups in 2D space, where fluxes would be associated to (dual) edges. In this case flux and holonomy representation would be self--dual to each other (for finite Abelian groups), reflecting the well--known weak--strong coupling duality for 2D statistical (Ising like) models, see for instance \cite{finite}. Whereas the usual LQG vacuum based on projective maps for the holonomies leads to a vacuum underlying the strong coupling limit, projections representing the coarse graining of fluxes could lead to a vacuum underlying the weak coupling limit, see also \cite{ccd}. This could be especially interesting for spin foam quantization,  as it is based on $BF$ theory, which is the weak coupling limit of lattice gauge theory.

\item[Loop quantum cosmology.] Loop quantum cosmology is the quantization of symmetry reduced models of classical general relativity using the methods of loop quantum gravity \cite{Bojowald:2008zzb,Ashtekar:2006rx,Ashtekar:2003hd}. The classical configuration space is the real line $\Rl$, while the quantum configuration space is given by an extension of the real line to what is called the Bohr compactification of the real line $\overline{\Rl}_{\text{Bohr}}$. This space can be given several independent descriptions. In particular, it can be understood as the Gel'fand spectrum of the algebra of almost periodic functions, which plays the role of the algebra of cylindrical functions for LQC. On the other hand, $\overline{\Rl}_{\text{Bohr}}$ can also be given a projective limit description \cite{Velhinho:2007gg}. 

Let us briefly recall this construction. For arbitrary $n\in\Nl$ consider the set of algebraically independent real numbers $\gamma=\{\mu_1,\ldots,\mu_n\}$, that is
\bas
\sum_{i=1}^n m_i\mu_i=0\,,\q m_i\in\Zl\q\Rightarrow\q m_i=0\q\forall\, i\,.
\eas
Consider now the subgroups of $\Rl$ freely generated by the set $\gamma$
\bas
G_\gamma:=\left\{\sum_{i=1}^n m_i\mu_i\,,\; m_i\in\Zl\right\}\,.
\eas
This induces a partial order on the set of all $\gamma$'s: $\gamma\prec \gamma'$ if $G_\gamma$ is a subgroup of $G_{\gamma'}$. The label set used in describing the projective structure of $\overline{\ca}$ consists of subgroupoids of $\cp$ generated by finite collections of holonomically independent edges. Here, the label set is exactly the set of all $\gamma$'s: collections of real numbers on a discrete real line. The projective structure of $\overline{\Rl}_{\text{Bohr}}$ is now constructed by defining
\bas
\Rl_\gamma:=\text{Hom}(G_\gamma, \U(1))\,,
\eas
and the surjective projections
\bas
p_{\gamma\gamma'}: \Rl_{\gamma'}\rightarrow \Rl_\gamma\,,\q \gamma\prec\gamma'\,.
\eas
Since $G_\gamma$ is freely generated by the set $\gamma=\{\mu_1,\ldots,\mu_n\}$ we can actually identify $\Rl_\gamma$ with $\U(1)^{n}$.
Finally, the family $\{\Rl_\gamma\}_\gamma$ forms a compact projective family, and its projective limit is homeomorphic to $\overline{\Rl}_{\text{Bohr}}$.

By definition the momentum space for LQC is $\Rl$ with discrete topology. Since each of the objects $\Rl_\gamma$ is an abelian group, we are in the setting of Theorem \ref{thm:proj-ind}. Thus, one may identify $\Rl$ with discrete topology with the inductive limit of the duals $\widehat{\Rl_\gamma}=\Zl^{n}$. We easily note the similarity of LQC and $\U(1)$-LQG from before, the subtlety being the index set over which the inductive/projective limit is taken. 

In \cite{Brunnemann:2007du} it was shown that the configuration space of LQC is not embeddable in the one of ($\SU(2)$) LQG. 
The natural question to ask, then, in view of our results, is whether one can instead embed LQC into $\U(1)$-LQG.

\end{description}

\acknowledgments{We would like to thank Aristide Baratin for discussions at an earlier stage of this project, Matti Raasakka for general discussions and in particular for correcting a flaw on theorem \ref{thm:imfz}, and Johannes Tambornino for useful comments on an earlier draft of this article. CG is supported by the Portuguese Science Foundation (\emph{Funda\c{c}\~ao para a Ci\^encia e a Tecnologia}) under research grant SFRH/BD/44329/2008, which he greatly acknowledges. 
Research at Perimeter Institute is supported by the Government of Canada through Industry Canada and by
the Province of Ontario through the Ministry of Research and Innovation. DO acknowledges support from the A. von Humboldt Stiftung through a Sofja Kovalevskaja Prize.
}

\appendix

\section{The $\U(1)$ group Fourier transform and the $\star$-product}
\label{sec:gft}
Let $f\in C(\U(1))$, that is, $f$ is a function of the form 
\bas
f: \U(1) & \rightarrow \Cl\\
\phi & \mt f(\phi)\,,
\eas
with $-\pi<\phi\leq\pi$ ($-\pi$ and $\pi$ obviously identified), and pointwise multiplication
\bas
(fg)(\phi)=f(\phi)g(\phi)\q \text{for all}\q f,g\in C(\U(1))\,.
\eas

The following theorem allows one to move from $C(\U(1))$ to $L^p(\U(1))$ which is much more structured.
\begin{theorem}[\cite{rudin}]
\label{thm:denseLp}
Let $X$ be a locally compact metric space and let $\mu$ be a $\sig$-finite regular Borel measure. Then the set $C_c(X)$ of continuous functions with compact support is dense in $L^p(X,\d\mu)$, $1\leq p<\infty$.
\end{theorem}
Since $\U(1)$ is compact, $C(\U(1))=C_c(\U(1))$ and we are done. $L^2$ are the only spaces of this class which are Hilbert spaces. Since we want to do quantum mechanics, we will stick to $L^2(\U(1))$. The inner product is 
\bas
\bra f,h\ket=\frac{1}{2\pi}\int_{-\pi}^\pi \d\phi\,\overline{f(\phi)}\,h(\phi)\,,\q \text{for all}\q f,h\in L^2(\U(1))\,.
\eas

Introduce the map $\cf$ for any $f$ in $L^2(\U(1))$ by
\ba
\cf(f)(x)=\frac{1}{2\pi}\int_{-\pi}^\pi \d\phi\, f(\phi)\,\e_\phi(x)=\frac{1}{2\pi}\int_{-\pi}^\pi \d\phi\, f(\phi)\,e^{-i\phi x}\,,
\label{ftu1}
\ea
where $\e_\phi(x)$ are the usual plane waves but for $x\in\Rl$. 
If $x\in\Zl$ we would have the usual Fourier transform.
The $\Im\cf$ is a certain set of continuous functions on $\Rl$, but certainly not all functions in $C(\Rl)$ are hit by $\cf$.  The inverse transformation reads
\bas
f(\phi)=\sum_{x\in\Zl}\cf(f)(x)e^{i\phi x}\,
\eas
and converges pointwise.\footnote{
Set 
\bas
S_N^f(\phi)=\int_{-\pi}^\pi \d\phi\, f(\phi)\, D_N(\phi)\,,
\eas
with $D_N(\phi)$ the Dirichlet kernel, i.e., $D_N(\phi)=\frac{1}{2\pi}\sum_{-N}^Ne^{i\phi x}$. Then, to prove that $\lim_{N\rightarrow\infty}S_N^f(\phi)=f(\phi)$, just use the fact that the Fourier coefficients of $S_N^f(\phi)-f(\phi)$ tend to zero as $N\rightarrow\infty$.} 
Notice here already that only the values of $\cf(f)$ on the integers are necessary to reconstruct back $f$.

Now, instead of the usual pointwise multiplication, we equip $\Im\cf$ with a $\star$-product. It is defined at the level of plane waves as
\bas
(\e_\phi\star\e_{\phi'})(x):=\e_{[\phi+\phi']}(x)
\eas
and extended to $\Im\cf$ by linearity. Here $[\phi+\phi']$ is the sum of two angles modulus $2\pi$ such that $-\pi<[\phi+\phi']\leq\pi$.

Given $u=\cf(f)$ and $v=\cf(h)$ we have explicitly
\ba
(u \star v)(x) &=  \int_{-\pi}^\pi \frac{\d\phi}{2\pi} \int_{-\pi}^\pi \frac{\d\phi'}{2\pi} f(\phi) h(\phi'-\phi)\, e^{-i\phi'x}\nonumber\\
&= \sum_{x',x'' \in \Zl}  u(x')\, v(x'') \frac{\sin(\pi(x'-x''))}{\pi(x'-x'')}\frac{\sin(\pi(x''-x))}{\pi(x''-x)}\nonumber\\
&=\sum_{x'\in \Zl}  u(x')\,v(x')\frac{\sin(\pi(x'-x))}{\pi(x'-x)}\,,
\label{eq:star_product}
\ea
where for the last line we used that 
\bas
\frac{\sin(\pi(x'-x''))}{\pi(x'-x'')} =\delta_{x',x''}\,,\q\text{whenever}\q x',x''\in\Zl\,.
\eas
Then, $\star$ is still commutative for $\U(1)$.

In order to have a better characterization of $\Im\cf$, we will give it a norm
\ba
||u||:=\sup_{x\in\Zl}|u(x)|\,,
\label{norm_u1}
\ea
and an involution $^*$, $u^*:=\bar{u}$, i.e. complex conjugation.

\begin{theorem}
$\Im\cf$ with the star product \eqref{eq:star_product}, norm \eqref{norm_u1} and complex conjugation as involution is a non-unital abelian $C^*$-algebra.
\end{theorem} 
\begin{proof}
First of all, we have to check that \eqref{norm_u1} is indeed a norm. We easily verify $||\a u||=|\a|\,||u||$, $||u+v||\leq ||u|| + ||v||$, for all $u,v\in\Im\cf$, $\a\in\Cl$. To see positive definiteness, notice that the functions on $\Im\cf$ are already determined by the values of $x\in\Zl$,
\bas
u(n)=\frac{1}{2\pi}\int & \d\phi\,f(\phi)\,e^{-i\phi n}=0\q \forall n\in \Zl\\
& \Rightarrow\q f=0\q (\text{almost everywhere})\q\Rightarrow\q u(x)=0\q \forall x\in\Rl\,.
\eas

The $C^*$-identity holds,
\bas
||u^*\star u||^2 = \sup_{x\in\Zl}|(u^*\star u)(x)|=\sup_{x\in\Zl}|(\bar{u}\cdot u)(x)|=||u||^2\,,
\eas
since for $x\in\Zl$ Eq.\eqref{eq:star_product} reduces to the usual pointwise product.

Finally, we show that $\Im\cf$ is complete in this norm.
Let $u_\a\in\Im\cf$ be a Cauchy sequence, that is, $u_\a$ is of the form
\bas
u_\a(x)=\frac{1}{2\pi}\int_{-\pi}^\pi \d\phi\, f_\a(\phi)\,e^{-i\phi x}\,,
\eas
for some $f_\a\in L^2(\U(1))$. We see that $u_\a$ is Cauchy iff $f_\a$ is Cauchy. Since $L^2(\U(1))$ is complete, $f_\a$ converges to some $f\in L^2(\U(1))$. This means that $u_\a$ converges to some $u$ of the form
\bas
u(x)=\frac{1}{2\pi}\int_{-\pi}^\pi \d\phi\, f(\phi)\,e^{-i\phi x}\,,
\eas
that is, $u\in\Im\cf$. Therefore, $\Im\cf$ is complete as well.

Using \eqref{eq:star_product} it is easy to convince ourselves that we have no unit function on $\Im\cf$. Furthermore, we have no constant functions at all. For instance, $u(x)=1$ does not live on $\Im\cf$ since it would correspond to $f(\phi)=2\pi\delta(\phi)$, a distribution. Distributions do not belong to $L^2(\U(1))$, unless we extend the framework to rigged Hilbert spaces. On the other hand, \eqref{eq:star_product} would give us
\ba
u(x)=\sum_{x'\in\Zl}u(x')\frac{\sin(\pi(x'-x))}{\pi(x'-x)}\,.
\label{lhsrhs}
\ea
Indeed, the $\star$-product is invariant under this transformation, that is, the $\star$-product does not see the difference between the l.h.s. and the r.h.s. of \eqref{lhsrhs}, as on the integers they coincide.

Thus, $\Im\cf$ is a non-unital abelian $C^*$-algebra.
\end{proof}
 
Using the Gel'fand representation theorem we will now prove that $\Im\cf\simeq C_0(\Zl)$.
\begin{theorem}[Gel'fand representation, \cite{gelfand}]
\label{thm:gelfand}
 Let $\fA$ be a (non-unital) commutative $C^*$-algebra. Then $\fA$ is isomorphic to the algebra of continuous functions that vanish at infinity over the locally compact Hausdorff space $\Delta(\fA)$ (the spectrum of $\fA$), $C_0(\Delta(\fA))$.
\end{theorem}

It remains to calculate the spectrum of $\Im\cf$, that is, the set of all non-zero $^*$-homomorphisms $\chi:\,\Im\cf\rightarrow \Cl$.
\begin{theorem}
\label{thm:imfz}
The spectrum of $\Im\cf$ is homeomorphic to $\Zl$, $\Delta(\Im\cf)\simeq\Zl$.
\end{theorem}
\begin{proof}
The $^*$-homomorphisms are
\bas
\chi_x:\Im\cf &\rightarrow \Cl\\
u & \mt \chi_x(u):=u(x)\,.
\eas
Clearly, $\chi_x(u\star v)=(u\star v)(x)=u(x)v(x)=\chi_x(u)\chi_x(v)$, and $\chi_x(u^*)=u^*(x)=\overline{u(x)}=\overline{\chi_x(u)}$.

We have to show that 
\bas
\cx:\Zl\rightarrow \Delta(\Im\cf)
\eas
is a homeomorphism (continuous bijection with continuous inverse). Define $\cx(x):=\chi_x$.

\emph{Continuity of $\cx$:}
 let $(x^\a)$ be a net in $\Zl$ converging to $x$, and let  $u\in \Im\cf$. First of all, notice that $u(x^\a)\rightarrow u(x)$ by continuity of the plane waves. Then, 
\bas
\lim_\a[\cx(x^\a)](u)=[\cx(x)](u)\q \Longleftrightarrow\q \lim_\a \check{u}(\cx(x^\a))=\check{u}(\cx(x))
\eas
for all $u\in\Im\cf$, hence $\cx(x^\a)\rightarrow \cx(x)$ in the Gel'fand topology.\footnote{The Gel'fand isometric isomomorphism is $\check\, : \Im\cf\rightarrow C_0(\Delta(\Im\cf))$, $u\mt\check{u}$ where $\check{u}(\chi)=\chi(u)$ where the space of continuous functions on the spectrum is equipped with the sup-norm. The Gel'fand topology is the weakest topology for which all the $\check{u}\in C_0(\Delta(\Im\cf))$ are continuous.} Actually, any function on $\Zl$ is continuous since the only topology available is the discrete one.

\emph{Injectivity:}
suppose $\cx(x)=\cx(x')$, then in particular $[\cx(x)](u)=[\cx(x')](u)$ for all $u\in\Im\cf$. We want to show that that $x=x'$. With a bit of logic we can turn this statement into the much easier one: [$x\neq x'$] implies [$u(x)\neq u(x')$ for some $u\in\Im\cf$]. Just pick $u=(1,0,\ldots)$. Thus, $\Im\cf$ separates the points of $\Zl$.

\emph{Surjectivity:}
 let $\chi\in\text{Hom}(\Im\cf,\Cl)$ be given. We must construct $x^\chi\in\Zl$ such that $\cx(x^\chi)=\chi$. Since $\Zl$ is a locally compact Hausdorff space, it is the spectrum of the abelian, non-unital $C^*$-algebra $C_0(\Zl)$, hence $\Zl=\text{Hom}(C_0(\Zl),\Cl)$. It follows that there exists $x^\chi\in \Zl$ such that $\chi(u)=u(x^\chi)$ for all $u\in C_0(\Zl)$.

\emph{Continuity of $\cx^{-1}$:}
let $(\chi^\a)$ be a net in $\Delta(\Im\cf)$ converging to $\chi$, so $\chi^\a(u)\rightarrow \chi(u)$ for any $u\in\Im\cf$. Then $\cx^{-1}(\chi^\a)\rightarrow \cx^{-1}(\chi)$.

Therefore, $\Im\cf\simeq C_0(\Zl)$, where $\Zl$ is endowed with the discrete topology.
\end{proof}

We now have $\Im\cf=C_0(\Zl)$. Once again we want to map $C_0(\Zl)$ to the much nicer (Hilbert) space $L^2(\Zl)=\ell^2(\Zl)$.
There are two ways of doing this. The first starts by noticing that $\overline{C_c(X)}=C_0(X)$, that is the space of continuous functions with compact support is dense in the space of continuous functions that vanish at infinity. Thus, any function in $C_0(X)$ can be approximated by functions in $C_c(X)$. Using this fact and Theorem \ref{thm:denseLp} we translate $C_0(\Zl)$ naturally to $\ell^2(\Zl)$ with the inner product
\bas
\langle u, v\rangle_\cf=\sum_{x\in\Zl}\overline{u}(x)\, v(x)\,,\q \text{for all}\q u,v\in\ell^2(\Zl)\,.
\eas

The second method uses the GNS construction. However, this relies on the choice of a state which, of course, can be chosen appropriately to get $\ell^2(\Zl)$. Let us then define a state for $u\in C_0(\Zl)$ as
\bas
\w_\cf(u):=\sum_{x\in\Zl} u(x)\,.
\eas
The induced inner product is
\bas
\bra u,v\ket_\cf &:=\w_\cf(\overline{u}\,v)=\sum_{x\in\Zl}\overline{u}(x)\, v(x)\,
\eas
for all $u,v\in C_0(\Zl)$, and the GNS Hilbert space is just $\ell^2(\Zl)$.

Finally, this characterization of $\Im\cf$ upgrades $\cf$ to an unitary transformation between $L^2(\U(1))$ and $\ell^2(\Zl)$
\bas
\bra \cf(f),\cf(h)\ket_\cf &= \sum_{x\in\Zl}\overline{u}(x)\, v(x)\\
&=\frac{1}{2\pi}\int_{-\pi}^\pi \d\phi\,\overline{f(\phi)}\,h(\phi)=\bra f, h\ket\,,
\q \text{for all}\q f,h\in L^2(\U(1))\,.
\eas
Hence, $\cf$ is just the usual Fourier transform on $\U(1)$.

\section{Projective and inductive limits}
\label{app:projind}
Let $\cl$ be a partially ordered and directed set, that is, we have a reflexive, antisymmetric, and transitive binary relation $\prec$ on the set $\cl$ such that for any $\gamma,\gamma'\in\cl$ there exists a $\gamma''\in\cl$ satisfying $\gamma,\gamma'\prec\gamma''$.

\subsection{Inverse or projective limits}
Let $\cc$ be a category. An inverse system in $\cc$ is a triple $(\cl, \{X_\gamma\},\{p_{\gamma\gamma'}\})$, where $\cl$ is a directed poset, $\{X_\gamma\}_{\gamma\in\cl}$ a collection of objects of $\cc$, and $p_{\gamma\gamma'}$ with $\gamma\prec\gamma'$ morphisms (projections) $p_{\gamma\gamma'}:X_{\gamma'}\rightarrow X_{\gamma}$ satisfying
\begin{itemize}
\item[(i)] $p_{\gamma\gamma}=\text{id}_{X_\gamma}$ for all $\gamma\in\cl$,
\item[(ii)] $p_{\gamma\gamma'}\circ p_{\gamma'\gamma''}=p_{\gamma\gamma''}$ whenever $\gamma\prec\gamma'\prec\gamma''$.
\end{itemize}

An object $X\in\text{Ob}(\cc)$ is called an inverse or projective limit of the system $(\cl, \{X_\gamma\},\{p_{\gamma\gamma'}\})$ and denoted $\varprojlim X_\gamma$, if there exist morphisms $p_\gamma:X\rightarrow X_\gamma$ for $\gamma\in\cl$ such that
\begin{itemize}
\item[(i)] for any $\gamma\prec\gamma'$ the diagram
\bas
\begin{CD} 
X @= X\\ 
@VV p_{\gamma'} V @VV p_\gamma V\\ 
X_{\gamma'} @> p_{\gamma\gamma'}>> X_\gamma\\
\end{CD}
\eas
commutes;
\item[(ii)] for any other $Y\in\text{Ob}(\cc)$ and morphisms $\pi_\gamma:Y\rightarrow X_\gamma$ with commuting diagram
\bas
\begin{CD} 
Y @= Y\\ 
@VV \pi_{\gamma'} V @VV \pi_\gamma V\\ 
X_{\gamma'} @> p_{\gamma\gamma'}>> X_\gamma\\
\end{CD}
\eas
for $\gamma\prec\gamma'$, there exists an unique morphism $m:Y\rightarrow X$ such that the following diagram
\bas
\begin{CD} 
Y @> m >> X\\ 
@VV \pi_{\gamma} V @VV p_\gamma V\\ 
X_{\gamma} @= X_\gamma\\
\end{CD}
\eas
commutes.
That is, if the inverse limit exists, it is unique up to $\cc$-isomorphism.
\end{itemize}

Finally, we remark that inverse limits admit the following description:
\ba
X\equiv \varprojlim X_\gamma=\{(x_\gamma)\in \prod X_\gamma \q |\q p_{\gamma\gamma'}(x_{\gamma'})=x_{\gamma}\q \text{for all}\q \gamma\prec\gamma'\}\,.
\ea

\subsection{Direct or inductive limits}
Let $\cc$ be a category. A direct system in $\cc$ is a triple $(\cl, \{X_\gamma\},\{\iota_{\gamma'\gamma}\})$, where $\cl$ is a directed poset, $\{X_\gamma\}_{\gamma\in\cl}$ a collection of objects of $\cc$, and $\iota_{\gamma'\gamma}$ with $\gamma\prec\gamma'$ morphisms (injections) $\iota_{\gamma'\gamma}:X_{\gamma}\rightarrow X_{\gamma'}$ satisfying
\begin{itemize}
\item[(i)] $\iota_{\gamma\gamma}=\text{id}_{X_\gamma}$ for all $\gamma\in\cl$,
\item[(ii)] $\iota_{\gamma''\gamma'}\circ \iota_{\gamma'\gamma}=\iota_{\gamma''\gamma}$ whenever $\gamma\prec\gamma'\prec\gamma''$.
\end{itemize}

An object $X\in\text{Ob}(\cc)$ is called a direct or inductive limit of the system $(\cl, \{X_\gamma\},\{\iota_{\gamma'\gamma}\})$ and denoted $\varinjlim X_\gamma$, if there exist morphisms $\iota_\gamma:X_{\gamma}\rightarrow X$ for $\gamma\in\cl$ such that
\begin{itemize}
\item[(i)] for any $\gamma\prec\gamma'$ the diagram
\bas
\begin{CD} 
X @= X\\ 
@AA \iota_{\gamma} A @AA \iota_{\gamma'} A\\ 
X_{\gamma} @> \iota_{\gamma'\gamma}>> X_{\gamma'}\\
\end{CD}
\eas
commutes;
\item[(ii)] for any other $Y\in\text{Ob}(\cc)$ and morphisms $i_\gamma:X_\gamma\rightarrow Y$ with commuting diagram
\bas
\begin{CD} 
Y @= Y\\ 
@AA  i_{\gamma} A @AA i_{\gamma'} A\\ 
X_{\gamma} @> \iota_{\gamma'\gamma}>> X_{\gamma'}\\
\end{CD}
\eas
for $\gamma\prec\gamma'$, there exists an unique morphism $m:X\rightarrow Y$ such that the following diagram
\bas
\begin{CD} 
X_\gamma @= X_\gamma\\ 
@VV \iota_{\gamma} V @VV i_\gamma V\\ 
X @> m >> Y\\
\end{CD}
\eas
commutes.
That is, if the direct limit exists, it is unique up to $\cc$-isomorphism.
\end{itemize}

We remark here that we may define the inductive limit differently. Let $\sim$ be the following relation on $\cup X_\gamma$: for $x\in X_\gamma$ and $y\in X_{\gamma'}$, then $x\sim y$ if there exists $\gamma''\in\cl$ such that $\iota_{\gamma''\gamma}(x)=\iota_{\gamma''\gamma'}(y)$ (identifying each $X_\gamma$ with its image in $\cup X_\gamma$). Since $\cl$ is a directed set, $\sim$ is an equivalence relation and one can show that
\ba
X\equiv \varinjlim X_\gamma =\cup_\gamma X_\gamma / \sim\,.
\ea

\bibliographystyle{JHEP}
\bibliography{thebibliography}

\providecommand{\href}[2]{#2}\begingroup\raggedright\begin{thebibliography}{10}

\bibitem{relativelocality}
G.~Amelino-Camelia, L.~Freidel, J.~Kowalski-Glikman, and L.~Smolin, {\it {The
  principle of relative locality}},  {\em Phys.Rev.} {\bf D84} (2011) 084010,
  [\href{http://xxx.lanl.gov/abs/1101.0931}{{\tt arXiv:1101.0931}}].

\bibitem{cogravity}
S.~Majid, {\it {Duality principle and braided geometry}},
  \href{http://xxx.lanl.gov/abs/hep-th/9409057}{{\tt hep-th/9409057}}.

\bibitem{TTbook}
T.~Thiemann, {\em Modern Canonical Quantum General Relativity}.
\newblock Cambridge University Press, 2007.

\bibitem{Freidel:2005me}
L.~Freidel and E.~R. Livine, {\it {Effective 3-D quantum gravity and
  non-commutative quantum field theory}},  {\em Phys.Rev.Lett.} {\bf 96} (2006)
  221301, [\href{http://xxx.lanl.gov/abs/hep-th/0512113}{{\tt
  hep-th/0512113}}].

\bibitem{Freidel:2005bb}
L.~Freidel and E.~R. Livine, {\it {Ponzano-Regge model revisited III: Feynman
  diagrams and effective field theory}},  {\em Class.Quant.Grav.} {\bf 23}
  (2006) 2021--2062, [\href{http://xxx.lanl.gov/abs/hep-th/0502106}{{\tt
  hep-th/0502106}}].

\bibitem{Freidel:2005ec}
L.~Freidel and S.~Majid, {\it {Noncommutative harmonic analysis, sampling
  theory and the Duflo map in 2+1 quantum gravity}},  {\em Class.Quant.Grav.}
  {\bf 25} (2008) 045006, [\href{http://xxx.lanl.gov/abs/hep-th/0601004}{{\tt
  hep-th/0601004}}].

\bibitem{Joung:2008mr}
E.~Joung, J.~Mourad, and K.~Noui, {\it {Three Dimensional Quantum Geometry and
  Deformed Poincare Symmetry}},  {\em J.Math.Phys.} {\bf 50} (2009) 052503,
  [\href{http://xxx.lanl.gov/abs/0806.4121}{{\tt arXiv:0806.4121}}].

\bibitem{etera}
E.~R. Livine, {\it {Matrix models as non-commutative field theories on R**3}},
  {\em Class.Quant.Grav.} {\bf 26} (2009) 195014,
  [\href{http://xxx.lanl.gov/abs/0811.1462}{{\tt arXiv:0811.1462}}].

\bibitem{matti}
M.~Raasakka, {\it {Group Fourier transform and the phase space path integral
  for finite dimensional Lie groups}},
  \href{http://xxx.lanl.gov/abs/1111.6481}{{\tt arXiv:1111.6481}}.

\bibitem{danielearistideSIGMA}
A.~Baratin and D.~Oriti, {\it {Noncommutative metric variables in loop quantum
  gravity, spin foams, and group field theory}},  {\em To appear in SIGMA}
  (2013).

\bibitem{carlosmattidaniele}
C.~Guedes, D.~Oriti, and M.~Raasakka, {\it {Quantization maps, algebra
  representation and non-commutative Fourier transform}},
  \href{http://xxx.lanl.gov/abs/1301.7750}{{\tt arXiv:1301.7750}}.

\bibitem{isham}
A.~Ashtekar and C.~Isham, {\it {Representations of the holonomy algebras of
  gravity and nonAbelian gauge theories}},  {\em Class.Quant.Grav.} {\bf 9}
  (1992) 1433--1468, [\href{http://xxx.lanl.gov/abs/hep-th/9202053}{{\tt
  hep-th/9202053}}].

\bibitem{Ashtekar:1993wf}
A.~Ashtekar and J.~Lewandowski, {\it {Representation theory of analytic
  holonomy C* algebras}},  \href{http://xxx.lanl.gov/abs/gr-qc/9311010}{{\tt
  gr-qc/9311010}}. In Quantum Gravity and Knots, ed. by J. Baez, Oxford Univ.
  Press.

\bibitem{Ashtekar:1994mh}
A.~Ashtekar and J.~Lewandowski, {\it {Projective techniques and functional
  integration for gauge theories}},  {\em J.Math.Phys.} {\bf 36} (1995)
  2170--2191, [\href{http://xxx.lanl.gov/abs/gr-qc/9411046}{{\tt
  gr-qc/9411046}}].

\bibitem{Ashtekar:1994wa}
A.~Ashtekar and J.~Lewandowski, {\it {Differential geometry on the space of
  connections via graphs and projective limits}},  {\em J.Geom.Phys.} {\bf 17}
  (1995) 191--230, [\href{http://xxx.lanl.gov/abs/hep-th/9412073}{{\tt
  hep-th/9412073}}].

\bibitem{Ashtekar:1998ak}
A.~Ashtekar, A.~Corichi, and J.~A. Zapata, {\it {Quantum theory of geometry.
  III: Non-commutativity of Riemannian structures}},  {\em Class. Quant. Grav.}
  {\bf 15} (1998) 2955--2972,
  [\href{http://xxx.lanl.gov/abs/gr-qc/9806041}{{\tt gr-qc/9806041}}].

\bibitem{Baratin:2010nn}
A.~Baratin, B.~Dittrich, D.~Oriti, and J.~Tambornino, {\it {Non-commutative
  flux representation for loop quantum gravity}},  {\em Class. Quant. Grav.}
  {\bf 28} (2011) 175011, [\href{http://xxx.lanl.gov/abs/1004.3450}{{\tt
  arXiv:1004.3450}}].

\bibitem{lewandowski}
M.~Bobienski, J.~Lewandowski, and M.~Mroczek, {\it {A Two surface quantization
  of the Lorentzian gravity}},
  \href{http://xxx.lanl.gov/abs/gr-qc/0101069}{{\tt gr-qc/0101069}}.

\bibitem{coherentfluxes}
D.~Oriti, R.~Pereira, and L.~Sindoni, {\it {Coherent states in quantum gravity:
  a construction based on the flux representation of LQG}},  {\em Journal of
  Physics A: Mathematical and Theoretical} {\bf 45} (2012), no.~24 244004,
  [\href{http://xxx.lanl.gov/abs/1110.5885}{{\tt arXiv:1110.5885}}].

\bibitem{coherentcollective}
D.~Oriti, R.~Pereira, and L.~Sindoni, {\it {Coherent states for quantum
  gravity: towards collective variables}},  {\em Classical and Quantum Gravity}
  {\bf 29} (2012), no.~13 135002,
  [\href{http://xxx.lanl.gov/abs/1202.0526}{{\tt arXiv:1202.0526}}].

\bibitem{johannesthomas}
K.~Giesel, J.~Tambornino, and T.~Thiemann, {\it {LTB spacetimes in terms of
  Dirac observables}},  {\em Class.Quant.Grav.} {\bf 27} (2010) 105013,
  [\href{http://xxx.lanl.gov/abs/0906.0569}{{\tt arXiv:0906.0569}}].

\bibitem{HannoTimSigmaReview}
T.~Koslowski and H.~Sahlmann, {\it {Loop quantum gravity vacuum with
  nondegenerate geometry}},  \href{http://xxx.lanl.gov/abs/1109.4688}{{\tt
  arXiv:1109.4688}}.

\bibitem{improved}
B.~Bahr and B.~Dittrich, {\it {Improved and Perfect Actions in Discrete
  Gravity}},  {\em Phys.Rev.} {\bf D80} (2009) 124030,
  [\href{http://xxx.lanl.gov/abs/0907.4323}{{\tt arXiv:0907.4323}}].

\bibitem{wroc}
B.~Bahr and B.~Dittrich, {\it {Breaking and restoring of diffeomorphism
  symmetry in discrete gravity}},
  \href{http://xxx.lanl.gov/abs/0909.5688}{{\tt arXiv:0909.5688}}.

\bibitem{song}
B.~Bahr, B.~Dittrich, and S.~He, {\it {Coarse graining free theories with gauge
  symmetries: the linearized case}},  {\em New J.Phys.} {\bf 13} (2011) 045009,
  [\href{http://xxx.lanl.gov/abs/1011.3667}{{\tt arXiv:1011.3667}}].

\bibitem{ccd}
B.~Dittrich, {\it {From the discrete to the continuous: Towards a cylindrically
  consistent dynamics}},  {\em New J.Phys.} {\bf 14} (2012) 123004,
  [\href{http://xxx.lanl.gov/abs/1205.6127}{{\tt arXiv:1205.6127}}].

\bibitem{aristidedanieleGFTmetric}
A.~Baratin and D.~Oriti, {\it {Group field theory with non-commutative metric
  variables}},  {\em Phys.Rev.Lett.} {\bf 105} (2010) 221302,
  [\href{http://xxx.lanl.gov/abs/1002.4723}{{\tt arXiv:1002.4723}}].

\bibitem{aristidedanieleHolst}
A.~Baratin and D.~Oriti, {\it {Group field theory and simplicial gravity path
  integrals: A model for Holst-Plebanski gravity}},  {\em Phys.Rev.} {\bf D85}
  (2012) 044003, [\href{http://xxx.lanl.gov/abs/1111.5842}{{\tt
  arXiv:1111.5842}}].

\bibitem{aristidedanieleGFTbarrett}
A.~Baratin and D.~Oriti, {\it {Quantum simplicial geometry in the group field
  theory formalism: reconsidering the Barrett-Crane model}},  {\em New J.Phys.}
  {\bf 13} (2011) 125011, [\href{http://xxx.lanl.gov/abs/1108.1178}{{\tt
  arXiv:1108.1178}}].

\bibitem{GFTdiffeos}
A.~Baratin, F.~Girelli, and D.~Oriti, {\it {Diffeomorphisms in group field
  theories}},  {\em Phys.Rev.} {\bf D83} (2011) 104051,
  [\href{http://xxx.lanl.gov/abs/1101.0590}{{\tt arXiv:1101.0590}}].

\bibitem{bubbles}
S.~Carrozza and D.~Oriti, {\it {Bounding bubbles: the vertex representation of
  3d Group Field Theory and the suppression of pseudo-manifolds}},  {\em
  Phys.Rev.} {\bf D85} (2012) 044004,
  [\href{http://xxx.lanl.gov/abs/1104.5158}{{\tt arXiv:1104.5158}}].

\bibitem{bubblesjackets}
S.~Carrozza and D.~Oriti, {\it {Bubbles and jackets: new scaling bounds in
  topological group field theories}},  {\em Journal of High Energy Physics}
  (2012) 1--42, [\href{http://xxx.lanl.gov/abs/1203.5082}{{\tt
  arXiv:1203.5082}}].

\bibitem{mattidaniele}
D.~Oriti and M.~Raasakka, {\it {Quantum Mechanics on SO(3) via Non-commutative
  Dual Variables}},  {\em Phys.Rev.} {\bf D84} (2011) 025003,
  [\href{http://xxx.lanl.gov/abs/1103.2098}{{\tt arXiv:1103.2098}}].

\bibitem{Varadarajan:2004ui}
M.~Varadarajan, {\it {The Graviton vacuum as a distributional state in
  kinematic loop quantum gravity}},  {\em Class.Quant.Grav.} {\bf 22} (2005)
  1207--1238, [\href{http://xxx.lanl.gov/abs/gr-qc/0410120}{{\tt
  gr-qc/0410120}}].

\bibitem{Bahr:2007xa}
B.~Bahr and T.~Thiemann, {\it {Gauge-invariant coherent states for Loop Quantum
  Gravity. I. Abelian gauge groups}},  {\em Class.Quant.Grav.} {\bf 26} (2009)
  045011, [\href{http://xxx.lanl.gov/abs/0709.4619}{{\tt arXiv:0709.4619}}].

\bibitem{Velhinho:2004sf}
J.~Velhinho, {\it {On the structure of the space of generalized connections}},
  {\em Int.J.Geom.Meth.Mod.Phys.} {\bf 1} (2004) 311--334,
  [\href{http://xxx.lanl.gov/abs/math-ph/0402060}{{\tt math-ph/0402060}}].

\bibitem{Dupuis:2011fx}
M.~Dupuis, F.~Girelli, and E.~R. Livine, {\it {Spinors and Voros star-product
  for Group Field Theory: First Contact}},  {\em Phys. Rev. D} {\bf 86} (2012)
  105034, [\href{http://xxx.lanl.gov/abs/1107.5693}{{\tt arXiv:1107.5693}}].

\bibitem{refSF1}
B.~Dittrich and F.~C. Eckert, {\it {Towards computational insights into the
  large-scale structure of spin foams}},  {\em Journal of Physics: Conference
  Series} {\bf 360} (2012), no.~1 012004,
  [\href{http://xxx.lanl.gov/abs/1111.0967}{{\tt arXiv:1111.0967}}].

\bibitem{refSF2}
B.~Dittrich, F.~C. Eckert, and M.~Martin-Benito, {\it {Coarse graining methods
  for spin net and spin foam models}},  {\em New J.Phys.} {\bf 14} (2012)
  035008, [\href{http://xxx.lanl.gov/abs/1109.4927}{{\tt arXiv:1109.4927}}].

\bibitem{refSF3}
F.~Markopoulou, {\it {Coarse graining in spin foam models}},  {\em
  Class.Quant.Grav.} {\bf 20} (2003) 777--800,
  [\href{http://xxx.lanl.gov/abs/gr-qc/0203036}{{\tt gr-qc/0203036}}].

\bibitem{refSF4}
R.~Oeckl, {\it {Renormalization of discrete models without background}},  {\em
  Nucl.Phys.} {\bf B657} (2003) 107--138,
  [\href{http://xxx.lanl.gov/abs/gr-qc/0212047}{{\tt gr-qc/0212047}}].

\bibitem{refSF5}
J.~A. Zapata, {\it {Continuum spin foam model for 3-d gravity}},  {\em
  J.Math.Phys.} {\bf 43} (2002) 5612--5623,
  [\href{http://xxx.lanl.gov/abs/gr-qc/0205037}{{\tt gr-qc/0205037}}].

\bibitem{refSF6}
J.~A. Zapata, {\it {Local gauge theory and coarse graining}},  {\em Journal of
  Physics: Conference Series} {\bf 360} (2012), no.~1 012054,
  [\href{http://xxx.lanl.gov/abs/1203.2306}{{\tt arXiv:1203.2306}}]. Based on
  talk given at Loops 11-Madrid.

\bibitem{procstar}
R.~E. Harti and G.~Luk\'acs, {\it {Bounded and unitary elements in
  pro-$C^*$-algebras}},  {\em Applied Categorical Structures} {\bf 14} (2006),
  no.~2 151--164, [\href{http://xxx.lanl.gov/abs/math/0511068}{{\tt
  math/0511068}}].

\bibitem{finite}
B.~Bahr, B.~Dittrich, and J.~P. Ryan, {\it {Spin foam models with finite
  groups}},  \href{http://xxx.lanl.gov/abs/1103.6264}{{\tt arXiv:1103.6264}}.

\bibitem{Bojowald:2008zzb}
M.~Bojowald, {\it {Loop quantum cosmology}},  {\em Living Rev.Rel.} {\bf 11}
  (2008) 4.

\bibitem{Ashtekar:2006rx}
A.~Ashtekar, T.~Pawlowski, and P.~Singh, {\it {Quantum nature of the big
  bang}},  {\em Phys.Rev.Lett.} {\bf 96} (2006) 141301,
  [\href{http://xxx.lanl.gov/abs/gr-qc/0602086}{{\tt gr-qc/0602086}}].

\bibitem{Ashtekar:2003hd}
A.~Ashtekar, M.~Bojowald, and J.~Lewandowski, {\it {Mathematical structure of
  loop quantum cosmology}},  {\em Adv.Theor.Math.Phys.} {\bf 7} (2003)
  233--268, [\href{http://xxx.lanl.gov/abs/gr-qc/0304074}{{\tt
  gr-qc/0304074}}].

\bibitem{Velhinho:2007gg}
J.~Velhinho, {\it {The Quantum configuration space of loop quantum cosmology}},
   {\em Class.Quant.Grav.} {\bf 24} (2007) 3745--3758,
  [\href{http://xxx.lanl.gov/abs/0704.2397}{{\tt arXiv:0704.2397}}].

\bibitem{Brunnemann:2007du}
J.~Brunnemann and C.~Fleischhack, {\it {On the configuration spaces of
  homogeneous loop quantum cosmology and loop quantum gravity}},
  \href{http://xxx.lanl.gov/abs/0709.1621}{{\tt arXiv:0709.1621}}.

\bibitem{rudin}
W.~Rudin, {\em Real and Complex Analysis}.
\newblock McGraw-Hill, 1970.

\bibitem{gelfand}
I.~Gelfand and M.~Neumark, {\it {On the imbedding of normed rings into the ring
  of operators in Hilbert space}},  {\em Rec. Math. [Mat. Sbornik] N.S} {\bf
  12(54)} (1943) 197--217.

\end{thebibliography}\endgroup

\end{document}